\documentclass[journal]{IEEEtran}

\usepackage{lipsum}
\usepackage{tikz} 
\usepackage{tikz-cd}
\tikzcdset{scale cd/.style={every label/.append style={scale=#1},
    cells={nodes={scale=#1}}}}
\usepackage{bbold}
\usepackage{amsthm}
\usepackage{nicefrac}
\usepackage{enumitem}
\usepackage{epsfig}
\usepackage{amsfonts}
\usepackage{amssymb}
\usepackage{mathrsfs}
\usepackage{amsmath}
\usepackage{times}
\usepackage{color}
\usepackage{soul}
\usepackage{mathtools}
\usepackage{float}
\usepackage{cite}
\usepackage{array}
\usepackage[caption=false,font=footnotesize]{subfig}
\usepackage{stfloats}
\usepackage{url}
\usepackage{enumitem}
\usepackage{algpseudocode}
\usepackage{algorithm}

\newtheorem{innercustomgeneric}{\customgenericname}
\providecommand{\customgenericname}{}
\newcommand{\newcustomtheorem}[2]{%
  \newenvironment{#1}[1]
  {%
   \renewcommand\customgenericname{#2}%
   \renewcommand\theinnercustomgeneric{##1}%
   \innercustomgeneric
  }
  {\endinnercustomgeneric}
}

\newcustomtheorem{customtheorem}{Theorem}
\newcustomtheorem{customprop}{Proposition}
\newcustomtheorem{customdefin}{Definition}
\newcustomtheorem{customcorollary}{Corollary}
\newcustomtheorem{customlemma}{Lemma}

\newtheorem{theorem}{Theorem}
\newtheorem{defin}{Definition}
\newtheorem{lemma}{Lemma}
\newtheorem{corollary}{Corollary}

\newtheorem{proposition}{Proposition}
\newcommand{\ox}[0]{\ensuremath{\otimes}}
\newcommand{\id}[0]{\ensuremath\mathbb 1}
\newcommand{\ff}[0]{\ensuremath{\mathbb{F}_2}}
 
\DeclareMathOperator{\im}{\ensuremath{\mathrm{Im}}}
\DeclareMathOperator{\rk}{\ensuremath{\mathrm{rank}}}
\newcommand{\cab}{\ensuremath{\mathcal{C}(\delta_A, \delta_B)}}

\newcommand{\wtrc}{\ensuremath{\mathrm{wt_{rc}}}}
\newcommand{\wtrcl}{\ensuremath{\mathrm{wt_{rc}^{log}}}}

\newcommand{\col}{\ensuremath{\mathrm{col}}}
\newcommand{\row}{\ensuremath{\mathrm{row}}}
\newcommand{\coll}{\ensuremath{\mathrm{col_{log}}}}
\newcommand{\rowl}{\ensuremath{\mathrm{row_{log}}}}

\begin{document}
\title{ReShape: a decoder for hypergraph product codes}
\author{Armanda~O.~Quintavalle
and Earl~T.~Campbell
\thanks{This work was supported by the Engineering and Physical Sciences Research Council [grant numbers EP/M024261/1 (E.T.C)].}%
\thanks{A. O. Quintavalle is with Department of Physics \& Astronomy, University of Sheffield, Sheffield, S3 7RH, United Kingdom (e-mail: armandaoq@gmail.com).}%
\thanks{E. T. Campbell is with Department of Physics \& Astronomy, University of Sheffield, Sheffield, S3 7RH, United Kingdom and Riverlane, Cambridge, CB2 3BZ, United Kingdom.
}%
}

\maketitle

\begin{abstract}
The design of decoding algorithms is a significant technological component in the development of fault-tolerant quantum computers.  Often design of quantum decoders is inspired by classical decoding algorithms, but there are no general principles for building quantum decoders from classical decoders.  Given any pair of classical codes, we can build a quantum code using the hypergraph product, yielding a hypergraph product code.  Here we show we can also lift the decoders for these classical codes.  That is, given oracle access to a minimum weight decoder for the relevant classical codes, the corresponding $[[n,k,d]]$ quantum code can be efficiently decoded for any error of weight smaller than $(d-1)/2$.  The quantum decoder requires only $O(k)$ oracle calls to the classical decoder and $O(n^2)$ classical resources. The lift and the correctness proof of the decoder have a purely algebraic nature that draws on the discovery of some novel homological invariants of the hypergraph product codespace. While the decoder works perfectly for adversarial errors, that is errors of weight up to half the code distance, it is not suitable for more realistic stochastic noise models and therefore can not be used to establish an error correcting threshold.
\end{abstract}

\IEEEpeerreviewmaketitle
The construction of quantum codes often takes classical codes as a starting point. The CSS construction is one method for combining a pair of classical codes into a quantum code.  However, the CSS recipe only works when the pair of classical codes are dual to each other.  Unfortunately, some of the best known classical code families, such as those based on expander graphs, do not come in convenient dual pairs.  The hypergraph product is a different recipe that allows a pair of arbitrary classical codes to form the basis of a quantum code~\cite{tillich2014quantum}.  Crucially, when the hypergraph product uses families of classical low-density parity check (LDPC) codes, it leads to families of quantum-LDPC codes.  The quantum-LDPC property eases the experimental difficulty of implementation and, combined with suitably growing distance, ensures the existence of an error correction threshold~\cite{Kovalevbadcode}.  

Two of the most widely known quantum codes, the toric and planar surface codes, are hypergraph product codes that use the classical repetition code as their seed classical code.  The decoding problem for the surface code can be recast as a minimum-weight perfect-matching problem, which is efficiently solved by the blossom algorithm~\cite{dennis02,FowlerRepMeasure} and the union-find algorithm \cite{delfosse2017almost}. Another interesting class of hypergraph product codes uses classical expander codes as their seed, with the resulting offspring called quantum expander codes~\cite{leverrier2015quantum}, which are quantum-LDPC codes achieving both constant rate and $\Omega(\sqrt{n})$ distance.  The classical expander codes can be decoded by a very simple bit-flip algorithm discovered by Spiser and Spielman~\cite{sipser1996expander}.  This inspired the small-set flip decoder for quantum expander codes, which follows a similar idea but is slightly modified, and has been shown to correct adversarial errors~\cite{leverrier2015quantum}, stochastic errors~\cite{fawzi2018} and also to operate as a single-shot decoder~\cite{leverrier18}. However, any binary linear code can be used as a seed to build hypergraph product codes. Using classical codes other than repetition and expander codes, for instance the semi-topological codes proposed in \cite{roffe2020decoding}, yield a broad range of hypergraph product codes for which there is no general propose decoder that is proven to work across the whole code family. 
% This is down to the fact that BP is an exact inference algorithm on trees but approximate on general graphs
For classical LDPC codes, using a belief propagation decoder (BP) works well in practice but it cannot be used out of the box on quantum-LDPC codes. In fact whenever a decoding instance has more than one minimum weight solution, it is degenerate, BP does not converge and yields a decoding failure. Degeneracy is the quintessential feature of quantum codes and therefore some workarounds are needed to use BP on quantum-LDPC codes \cite{poulin08, babar}. The literature offers many examples of BP inspired decoders for quantum-LDPC codes which show an error correcting threshold~\cite{ liu2019neural,rigby,panteleev2019degenerate, yoder20,grospellier2020combining, roffe2020decoding}, however none of them come with a correctness proof. Recently, a union-find like decoder has been proposed to decode quantum-LDPC codes \cite{delfosse2021union_ldpc}. The authors in \cite{delfosse2021union_ldpc} prove that their union-find decoder corrects for all errors of weight up to a polynomial in the distance for three classes of quantum-LDPC codes: codes with linear confinement (see \cite{bombin2015single,quintavalle2020single}), $D$-dimensional hyperbolic codes and $D$-dimensional toric codes for $D\ge 3$. The decoder in \cite{delfosse2021union_ldpc} is therefore provably correct for adversarial noise, nonetheless a comprehensive investigation of its performance under stochastic noise is still missing. 

Here we introduce the ReShape decoder for generic hypergraph product codes.  Given a $[[n,  k, d]]$ hypergraph product code built using classical codes with parity matrices $\delta_A$ and $\delta_B$, we assume access to a minimum weight decoder for parity matrices $\delta_A$, $\delta_B$, $\delta_A^T$ and $\delta_B^T$.   The ReShape decoder calls these classical decoders as blackbox oracles without any modification or knowledge of their internal working, and furthermore only requires $O(k)$ oracle calls, and only a polynomial amount of additional classical computation. Under these conditions we prove that ReShape works in the adversarial setting, correcting errors (up to stabilisers) of weight less than half the code distance.  Therefore, ReShape lifts the classical decoders to the status of a quantum decoder, providing the first general purpose hypergraph product codes decoder proven to correct adversarial errors. Formally we prove:
\begin{theorem}
\label{thm:reshape_success}
Any $[[n, k, d]]$ hypergraph product code constructed from the classical parity check matrices $\delta_A$ and $\delta_B$ can be successfully decoded from error of weight up to $(d-1)/2$ using $O(k)$ oracle calls to classical decoders for the seed matrices and their transpose plus $O(n^2)$ classical operations.
\end{theorem}
Theorem \ref{thm:reshape_success} though, does not state anything about stochastic noise or error correcting thresholds. Families of $n$-qubit hypergraph product codes have distance of at most $O(\sqrt{n})$ and so they are bad codes in the sense that the distance is sub-linear. However, given a stochastic noise model with each qubit affected independently with probability $p$, the typical error size will be $pn$. Thus, for $n > (d/2p)$, the most likely errors will not necessarily be corrected by ReShape and there is no guarantee that a threshold will be observed.  Indeed, we implemented ReShape for several code families and found evidence that ReShape fails to provide a threshold (see Figure \ref{fig:anti_threshold}).  A clear open problem is whether there exists a similar general lifting procedure, or modification of ReShape, for which one can prove good performance in the stochastic settings. Hence, if on one hand Theorem \ref{thm:reshape_success} provides a solution to the adversarial decoding problem for hypergraph product codes, on the other, a stronger, difficult and much longed-for result is desirable. Namely, the solution of the stochastic decoding problem for hypergraph product codes both on a theoretical level (proof of a threshold) and on a practical one (numerical observation of a high correcting threshold). Even so, ReShape still provides some improvement over state-of the-art BP and union-find like decoders for stochastic noise. First, ReShape comes with a proof of correctness, that BP lacks; second, the proof works for all errors up to the optimal value of $(d-1)/2$, whilst the modification of union-find proposed in \cite{delfosse2021union_ldpc} is provably correct only for errors of weight up to $A d^{\alpha}$, for some $A, \alpha>0$ and $\alpha < 1$.
% works }oppose to BP, it comes with a prove if we contrast ReShape with what said on BP and union-find like decoders, we can say that see that ReShape improves on the former in that it is provably correct and on the latter because its correction capability are optimal, $(d-1)/2$ against a polynomial in $d$.  

\section{Preliminaries and notation}
\label{sec:preliminaries}
A classical $[n, k, d]$ linear code is compactly described by its parity check matrix $H$. The matrix $H$ is a binary matrix of size $m \times n$ such that the codespace $\mathcal C(H) \subseteq \mathbb F_2^n$ is described by:
\begin{align}
\label{eq:classical_code}
    \mathcal{C}(H) = \{v \in \mathbb F_2^n &: Hv = 0\}.
\end{align}
The codespace $\mathcal{C}(H)$ has dimension $k = n - \rk(H)$ and distance
$d$ defined as:
\begin{align*}
    d = \min \{|v| &: v \in \mathcal C(H), v\neq 0\},
\end{align*}
where $|v|$ is the Hamming weight of the binary vector $v$. Whenever the parity check matrix has columns and rows of small weight we say that it is a low density parity check (LDPC) matrix; when $H$ has constant column and row weight $w_c, w_r$ we shortly say that it is a $(w_c, w_r)$-matrix.
 
The classical decoding problem can be stated as: given a \emph{syndrome} vector $s \in \mathbb F_2^m$, find the minimum weight solution $e \in \mathbb{F}_2^n$ to the equation \begin{align}
    \label{eq:classical_decoding}
    He = s.
\end{align}
It is easy to show that the optimal decoder for any classical linear code can correct errors of weight up to half the code distance (see, for instance, \cite{huffman2010fundamentals}). 

A quantum $[[n, k, d]]$ stabiliser code \cite{gottesman1997stabilizer} is a subspace of dimension $2^k$ of the Hilbert space $(\mathbb C^2)^{\ox n}$. It is described as the common $+1$ eigenspace of its stabiliser group $\mathcal S$, an Abelian subgroup of the Pauli group $\mathcal{P}_n$ such that $-\id \not \in \mathcal{S}$. The Pauli group on $n$ qubits is the group generated by the $n$-fold tensor product of single qubit Pauli operators. The weight $|P|$ of a Pauli operator $P \in \mathcal P_n$ is the number of its non-identity factors. We indicate by $\mathcal{N}(\mathcal S)$ the normaliser of $\mathcal S$ i.e.\ the group of Paulis which commute with the stabiliser group $\mathcal{S}$. Because $\mathcal{S} \subseteq \mathcal{N}(\mathcal{S})$, the quotient group 
\begin{align*}
\mathcal L =  \nicefrac{\mathcal{N(\mathcal S)}}{\mathcal{S}}  \end{align*}
is well defined and referred to as \emph{homology group}, see Appendix \ref{app:hpg_codes}. Elements $[P]$ of $\mathcal{L}$ are \emph{homology classes}: equivalence classes with respect to the congruence modulo multiplication by stabiliser operators. Explicitly:
\begin{align}
\label{eq:homology}
    [P] = \{PS &: S \in \mathcal{S}\},
\end{align}
and for any Pauli $P$, its homology class $[P]$ is uniquely defined via Eq.~\eqref{eq:homology}. Importantly, each Pauli $P$ such that $[P] \neq [\id]$ in $\mathcal L$ is an operator that preserves the codespace and has non-trivial action on it. We refer to such code operators modulo $\mathcal S$ as  \emph{logical Pauli operators}; with slight abuse of notation we write $P \in \mathcal L$, meaning $[P] \in \mathcal L$. Two logical operators $P, Q$ are said to be \emph{homologically equivalent}, or just equivalent, if and only if they belong to the same homology class i.e.\ by Eq.~\eqref{eq:homology}, if and only if $ [P] = [Q]$. Importantly, for a code of dimension $k$, $\mathcal L \simeq \mathcal{P}_k $. The distance $d$ of the code is the minimum weight of any non-trivial logical operator in $\mathcal L$. Any generating set of the stabiliser group $\mathcal S$ induces a syndrome map $\sigma$. Namely, if $\mathcal S = \langle S_1, \dots, S_m \rangle $, the associated syndrome function $\sigma$ maps any Pauli $P \in \mathcal{P}_n$ in a binary vector $s=(s_1, \dots, s_m)^T \in \mathbb F_2^m$ such that $s_i = 0$ if and only if $P$ commutes with $S_i$ and $1$ otherwise. We refer to the vector $s$ as the\emph{ syndrome}. Conventionally, when considering a stabiliser code, it is always intended that a generating set $\{S_1, \dots, S_m\}$ for the stabiliser group is chosen and with it a syndrome map. We say that a stabiliser code is LDPC if each $S_i$ has low weight and each qubit is in the support of only a few generators.

The decoding problem for stabiliser codes can be stated as: given a syndrome vector $s \in \mathbb F_2^m$, find an operator $E_r \in \mathcal{P}_n$ such that (i) $\sigma(E_r) = s$ and (ii) $[E_r] = [E_{\min}]$, where $E_{min}$ is a minimum weight operator with syndrome $s$. We call any operator that satisfies (i) a \emph{valid} solution of the syndrome equation and operators for which both (i) and (ii) are true, \emph{correct} solutions.

Pauli operators can be put into a one-to-one correspondence with binary vectors, if we discard the phase factor $\pm i$. In fact, any Pauli $P$ can be written as:
\begin{align*}
    P &\propto X[v] \cdot Z[w], \\
    &= X^{v_1} \ox \dots \ox X^{v_n} \cdot Z^{w_1} \ox \dots \ox Z^{w_n}, & v, w \in \mathbb F_2^n
\end{align*}
from which it follows:
\begin{align}
    (X[v] Z[w])(X[v']Z[w']) &= \pm X[v+v']Z[w+w'], \label{eq:sum_vectors}
\end{align}
and two operators commute if and only if 
\begin{align}
\langle v, w' \rangle + \langle v', w \rangle = 0 \mod 2
% \big\langle X[v] Z[w], X[v']Z[w']\big\rangle &= \langle v, w' \rangle + \langle v', w \rangle.
\label{eq:scalar_product}
\end{align}
and anti-commute otherwise.
This correspondence between binary vectors and Pauli operators is particularly handy when dealing with CSS codes \cite{calderbank1996good, steane1996multiple}. CSS codes are stabiliser codes for which the stabiliser group can be generated by two disjoint sets $\mathcal S_x$ and $\mathcal S_z$ of $X$ and $Z$ type operators respectively. If $\mathcal S_x = \{X[v_1], \dots X[v_{m_x}]\}$, $\mathcal S_z = \{Z[w_1], \dots, Z[w_{m_z}]\}$ and we define $H_X$ and $H_Z$ as the matrices whose rows are the $v_i$s and the $w_i$s respectively, then the commutation relation on the stabilisers generators translate in to the binary constraint $H_X H_Z^T = 0$. Using Eq.~\eqref{eq:scalar_product}, it is easy to show that the syndrome for a Pauli error $E = X[e_x]Z[e_z]$ is described by the two binary vectors $s_z = H_Z e_x$ and $s_x = H_X e_z$. Since these two linear equations are independent, we can treat the $X$-part and $Z$-part of the error separately. For CSS codes, we define the $X$-distance $d_x$ as the minimum weight of an operator $X[v]$ which commutes with all the stabilisers in $\mathcal{S}_z$ but does not belong to the group generated by $\mathcal S_x$. Note that the weight of an operator $X[v]$ equates the Hamming weight $|v|$ of the vector $v$. Therefore, combining Eq.~\eqref{eq:sum_vectors}, Eq.~\eqref{eq:scalar_product} and the definition of $d_x$, we shortly say that $d_x$ is the minimum weight of a vector $v$ in $\ker H_Z$ which does not belong to the row span of $H_X$, i.e.
\begin{align} d_x := \mathrm{min} \{ |v| : H_Z v = 0, v \not \in \im H_X^T \}
\end{align}
Similarly, $d_z$ is the minimum weight of a vector in $\ker H_X$ not in $\im H_Z^T$.  

The $Z$-error decoding problem for CSS code can be stated as: given a syndrome vector $s \in \mathbb F_2^{m_x}$, find a valid and correct solution $e \in \mathbb F_2^n$ to the equation:
\begin{align}
    \label{eq:css_decoding}
    H_X e = s,
\end{align}
where $e_r$ is valid if and only if $H_X e_r = s$ and it is correct if and only if it belongs to the homology class of the minimum weight operator with syndrome $s$. Because for an operator $Z[e]$ its weight equates the Hamming weight of the vector $e$, the $Z$-decoding problem for CSS codes can be reformulated exactly as done for the classical decoding problem in Eq.~\eqref{eq:classical_decoding}. Explicitly, given $s$, find the minimum weight solution to the linear equation $H_X e = s$. The $X$-decoding problem is derived from Eq.~\eqref{eq:css_decoding} by duality, exchanging the role of $X$ and $Z$. 

It goes without saying that, if any CSS code defines two classical parity check matrices, the converse is also true. Namely, starting from any two binary matrices $H_1, H_2$ such that $H_1 H_2^T = 0$, this defines a CSS code with $H_X = H_1$, $H_Z = H_2$. If the classical linear code with parity check $H_i$ has parameters $[n, k_i, d_i]$, the associated quantum code code has parameters $[[n, k_1 + k_2 - n, d_x, d_z]]$ where $d_x \ge d_2$ and $d_z \ge d_1$. A review on quantum codes can be found, for instance, in \cite{roffe2019quantum, NC01b}. \\

In this article we focus on a sub-class of CSS codes, the hypergraph product codes \cite{tillich14, bravyi14, audoux2015tensor, zeng2019higher}. We give a minimal description of these codes in Section \ref{sec:main} and we refer the reader to Appendix \ref{app:hpg_codes} for a more detailed presentation. We study some homology invariants for the logical operators of the hypergraph product codes in Section \ref{sec:algebraic_properties}. These invariants are the algebraic core upon which we design a decoder for these codes, the ReShape decoder. We prove that ReShape is an efficient and correct decoder for adversarial noise in Section \ref{sec:reshape}. We conclude with some consideration on the performance of ReShape under stochastic noise in Section \ref{sec:stochastic}.  
\section{Hypergraph product codes}
\label{sec:main}
We here present a bottom-up overview on hypergraph product codes. The purpose of this Section is dual: we both want to describe the hypergraph product codes with the least possible technical overhead and introduce the notation necessary to motivate and give an intuition for the results presented in Section \ref{sec:results}. We refer the reader interested in the homology theory approach to Appendix \ref{app:hpg_codes}.

The most well-known example of hypergraph product code is the toric code and its variations \cite{Kit03, bravyi1998quantum}. The toric code is conventionally represented by a square lattice where qubits sit on edges, $X$-stabilisers are identified with vertices and $Z$-stabilisers with faces. Since a square lattice has two kind of edges, vertical and horizontal edges, the first evident feature of this identification is that, accordingly, there are two type of qubits. The second is that each vertex/$X$-stabiliser uniquely identifies a row of horizontal edges and a column of vertical one, starting from the four ones that are incident to it. The third is that faces/$Z$-stabilisers, similarly to vertices, uniquely identify a column of horizontal edges and a row of vertical ones, starting from the four which lie on its boundary. Very similar attributes can be found in all the hypergraph product codes, as we now explain.

Consider two classical parity check matrices $\delta_A, \delta_B$ of size $m_a \times n_a$ and $m_b \times n_b$; we indicate with $\cab$ their hypergraph product code and refer to the matrices $\delta_A$ and $\delta_B$ as \emph{seed} matrices. The qubits of the code $\cab$ can be labelled as \emph{left} and \emph{right} qubits. Left qubits can be placed in a $n_a \times n_b$ grid and right qubits in a $m_a \times m_b$ grid, see Figure \ref{fig:qubits_grid}. Under this labelling, left and right qubits are uniquely identified by pair of indices $(j_a, j_b)$ and $(i_a, i_b)$ respectively, where $j_a, j_b$ vary among the column indices of $\delta_A, \delta_B$ while $i_a, i_b$ vary among their row indices. Given a pair $(L, R)$ of binary matrices, of size $n_a \times n_b$ and $m_a \times m_b$ respectively, we define the $Z$-operator:
\begin{align}
\label{eq:z_operator}
    Z(L, R) &=\left( \bigotimes_{j_a, j_b} Z^{L_{j_a, j_b}}\right) \otimes \left(\bigotimes_{i_a, i_b} Z^{R_{i_a, i_b}}\right),
\end{align}
and similarly for $X$-operators. We refer to $L$ as the left part of the operator and to $R$ as its right part.

The code $\cab$ has $m_a \times n_b$ $X$-stabiliser generators which can be indexed by $(i_a, j_b)$. The $X$-stabiliser $S^x{(i_a, j_b)}$ has support contained in the union of the $j_b$th column of left qubits and the $i_a$th row of right qubits. More precisely\footnote{Here and in the following, for a $m \times n$ matrix $\delta$ we indicate by $\delta_i \in \mathbb F_2^n$ the transpose of its $i$th row, and by $\delta^j \in \mathbb F_2^m$ its $j$th column.}, it acts as $X[(\delta_A)_{i_a}]$ on the left qubits located at column $j_b$ and as $X[(\delta_B)^{j_b}]$ on the right qubits located on row $i_a$, see Figure \ref{fig:xstab}. Using the $X$-version of Eq.~\eqref{eq:z_operator}, $S^x{(i_a, j_b)}$ is uniquely represented by the pair of matrices, $L=\delta_A^T E_{{i_aj_b}}$ and $R=E_{i_aj_b}\delta_B^T$,  so that
\begin{align*}
   S^x{(i_a, j_b)} := X(\delta_A^T E_{{i_aj_b}}, E_{i_aj_b}\delta_B^T) ,
\end{align*}
where $E_{i_a, j_b}$ is the all-zero $m_a \times n_b$ matrix but for the $(i_a, j_b)$th entry which is $1$. From the characteristic `cross' shape of the stabilisers generators  $S^x(i_a, j_b)$, it follows that if $(G_L, G_R)$ is an $X$-stabiliser for $\cab$, then (i) each column of $G_L$, as a vector in $\ff^{n_a}$, belongs to $\im \delta_A^T$ and (ii) each row of $G_R$, as a vector in $\ff^{m_b}$, belongs to $\im \delta_B$.

Similarly, $Z$-stabiliser generators are indexed by $(j_a, i_b)$ for $1\le j_a \le n_a$ and $1 \le i_b \le m_b$ and $S^z(j_a, i_b)$ is uniquely represented by the pair of matrices:
\begin{align*}
(L,R) =    (E_{j_ai_b}\delta_B, \delta_AE_{j_ai_b}),
\end{align*}
for $E_{j_ai_b}$ of size $n_a \times m_b$, with all entries $0$ but for the $(j_a, i_b)$th entry which is $1$.

The syndrome equation for hypergraph product codes can be derived combining Eq.~\eqref{eq:css_decoding} and the expression for the stabiliser generators. By Eq.~\eqref{eq:css_decoding}, the $i$th bit of the syndrome vector $s \in \ff^{m_x}$ equates the inner product between the $i$th $X$-stabiliser generator, which corresponds to the $i$th row of the matrix $H_X$, and the error vector. In the same way, by reshaping vectors into matrices (see Appendix \ref{sec:reshaping}), the $(i_a, j_b)$th bit of the syndrome matrix $S \in \ff^{m_a \times n_b}$ equates the inner product of the $(i_a, j_b)$th $X$-stabiliser generator and the error matrices $(L, R)$:
\begin{align*}
    (\delta_A)_{i_a} L + R (\delta_B)^{j_b} = S_{{i_a}, {j_b}},
\end{align*}
and by linearity:
\begin{align}
    \label{eq:1_se}\tag{SE}
    \sigma(L, R) \coloneqq \delta_AL + R \delta_B.
\end{align}
It is easy to show that any $Z$-stabiliser has trivial $X$-syndrome, which is equivalent to $X$-stabilisers and $Z$-stabilisers commuting. As a consequence, $\cab$ is a well-defined CSS code.

A minimal generating set of logical $Z$-operators for $\cab$ is given by:
\begin{align}
\label{eq:matrix_logical_z}
\mathcal{L}_z&\coloneqq \mathcal{L}_z^{\mathrm{left}} \cup \mathcal{L}_z^{\mathrm{right}}
\end{align}
where:
\begin{align*}
    \mathcal{L}_z^{\mathrm{left}} \coloneqq \big\{(L, 0) &: L = k_a \cdot e_{j_b}^T,\\
    &\quad k_a \text{ varies among a basis of } \ker \delta_A,\\ 
    &\quad e_{j_b} \text{ varies among a basis of } (\im \delta_B^T)^{\bullet},\\
    &\quad |e_{j_b}|=1\big\},
\end{align*}
and
\begin{align*}
    \mathcal{L}_z^{\mathrm{right}} \coloneqq \big\{(0, R) &: R = e_{i_a} \cdot \bar k_b^T,\\
    &\quad \bar k_b \text{ varies among a basis of } \ker \delta_B^T, \\
    &\quad e_{i_a} \text{ varies among a basis of } (\im \delta_A)^{\bullet},\\
    &\quad |e_{i_a}|=1
    \big\}.
\end{align*}
Here, given a vector space $V \subseteq \ff^n$, $V^{\bullet}$ denotes any space such that $V \oplus V^{\bullet} \simeq \ff^n$. In particular, the space $V^{\bullet}$ is in general different from the orthogonal complement $V^{\bot}$ of the space $V$, see Appendix \ref{app:linear_algebra} for details.
Similarly, a minimal generating set of logical $X$-operators is:
\begin{align}
\label{eq:matrix_logical_x}
    \mathcal L_x &\coloneqq \mathcal L_x^{\mathrm{left}} \cup \mathcal L_x^{\mathrm{right}}
\end{align}
where:
\begin{align*}
    \mathcal{L}_x^{\mathrm{left}} \coloneqq \big\{(L, 0) &: L = e_{j_a} \cdot k_b^T,\\
    &\quad k_b \text{ varies among a basis of } \ker \delta_B,\\
    &\quad e_{j_a} \text{ varies among a basis of } (\im \delta_A^T)^{\bullet},\\
    &\quad |e_{j_a}|=1 \big\},
\end{align*}
and
\begin{align*}
    \mathcal{L}_x^{\mathrm{right}} \coloneqq \big\{(0, R) &: R = \bar k_a \cdot e_{i_b}^T,\\
    &\quad \bar k_a \text{ varies among a basis of } \ker \delta_A^T,\\
    &\quad e_{i_b} \text{ varies among a basis of } (\im \delta_B)^{\bullet},\\
    &\quad |e_{i_b}|=1  \big\}.
\end{align*}
To sum up, the code $\cab$ is a CSS code with parameters $[[n, k, d_x, d_z]]$, where:
\begin{align*}
    n &= n_an_b+m_am_b\\
    % k &= (n_a - \rk(\delta_A))(n_b -\rk(\delta_B)) \\
    % &\quad + (m_a -\rk(\delta_A)) (m_b-\rk(\delta_B))\\
    k & = (n_a - \mathrm{rk}_a) (n_b - \mathrm{rk_b}) + (m_a - \mathrm{rk}_a) (m_b - \mathrm{rk_b}) \\
    d_x &= \min\{d_a^T, d_b\}\\
    d_z &= \min\{d_a, d_b^T\}
\end{align*}
for $\mathrm{rk}_{\ell} = \rk(\delta_{\ell})$ and $d_{\ell}$ (resp. $d_{\ell}^T$) distance of the classical code with parity check matrix $\delta_{\ell}$ (resp. $\delta_{\ell}^T$), $\ell = A, B$. By convention, we define the distance of the trivial code $\{0\}$ to be $\infty$. In particular, whenever one or both seed matrices (or transpose) are full rank, one of the summands in the expression for $k$ cancel out e.g. if $\delta_A$ or $\delta_B$ have full rank, then $k = (n_a - \mathrm{rk}_a)(n_b - \mathrm{rk}_b)$, $d_x = d_b$ and $d_z=d_a$.

The similarities in structure between general hypergraph product codes $\cab$ and the toric codes (with and without boundaries) should now be clear: the toric code with boundaries (resp.\ without) of lattice size $L$ is just the hypergraph product code $\mathcal{C}(\delta_L, \delta_L)$ where $\delta_L$ is the full-rank $L-1 \times L$ (resp.\ non-full-rank $L \times L$) parity check matrix of the classical $[L, 1, L]$ repetition code, e.g.\, for $L = 3$:
\begin{align}
\label{eq:repetition_code}
 \begin{pmatrix}
 1 & 1 & 0\\
 0 & 1 & 1
 \end{pmatrix}, &&\text{ resp. }&& \begin{pmatrix}
 1 & 1 & 0 \\
 0 & 1 & 1\\
 1 & 0 & 1
 \end{pmatrix}.
\end{align}
Left and right qubits correspond to vertical and horizontal edges; vertices and faces on the square lattice can be indexed in the natural way yielding the same stabiliser indexing of the general hypergraph product codes; string like (resp.\ loop like) logical operators correspond precisely to the left and right logical operators described above which have single column/single row support.

In what follows, we focus on $Z$-errors and their correction. With slight abuse of notation, we will  refer to pair of matrices $(L, R)$ as operators (and vice versa sometimes) where the identification is clear via Eq.~\eqref{eq:z_operator}. The corresponding results for $X$-errors are easily obtained by duality as per any CSS code. More precisely, by swapping the role of $X$ and $Z$ but also the role of rows and columns; alternatively, considering the code $\mathcal{C}(\delta_A^T, \delta_B^T)$, see Appendix \ref{app:hpg_codes}. 
\begin{figure*}[tb]
\centering\subfloat[]{\includegraphics[scale=1]{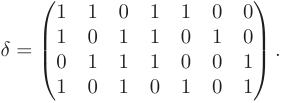}
\label{fig:delta}
}
\vfill
\subfloat[An $X$ stabiliser.]{
\includegraphics[scale=0.45]{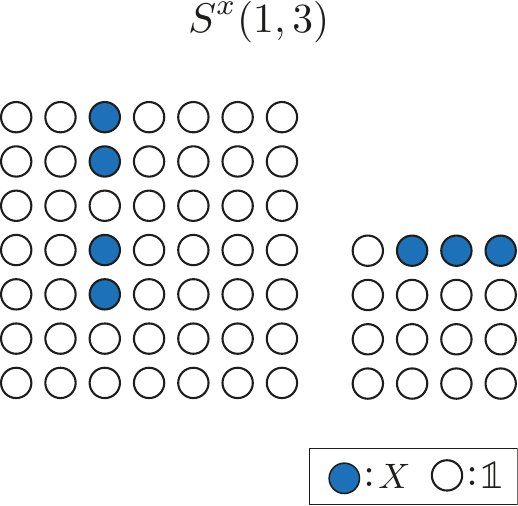}
\label{fig:xstab}
}\hfill
\subfloat[A logical left $X$ operator.]{
\includegraphics[scale=0.45]{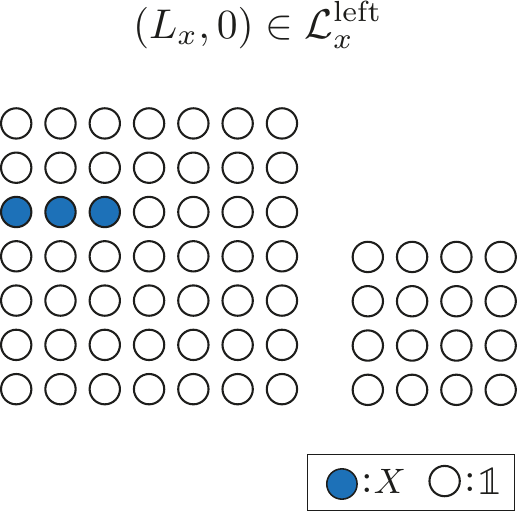}
\label{fig:xleft}
}\hfill
\subfloat[A logical right $X$ operator.]{
\includegraphics[scale=0.45]{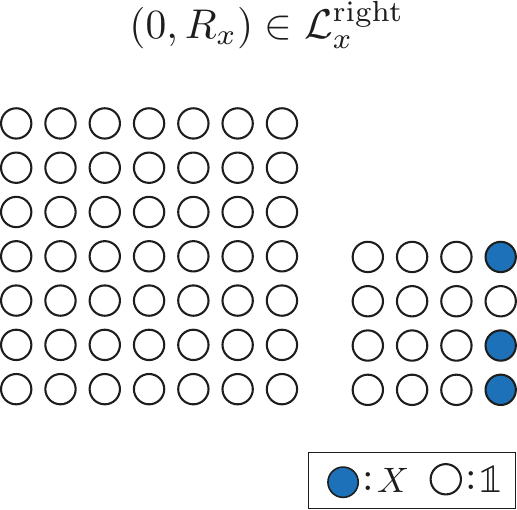}
\label{fig:xright}
}
\vfill
\subfloat[A $Z$ stabiliser.]{
\includegraphics[scale=0.45]{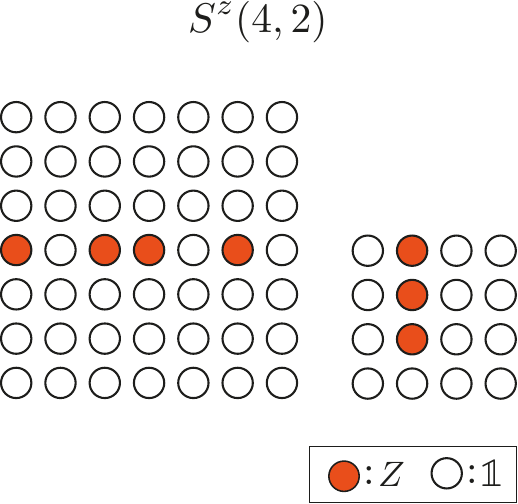}
\label{fig:zstab}
}
\hfill
\subfloat[A logical left $Z$ operator.]{
\includegraphics[scale=0.45]{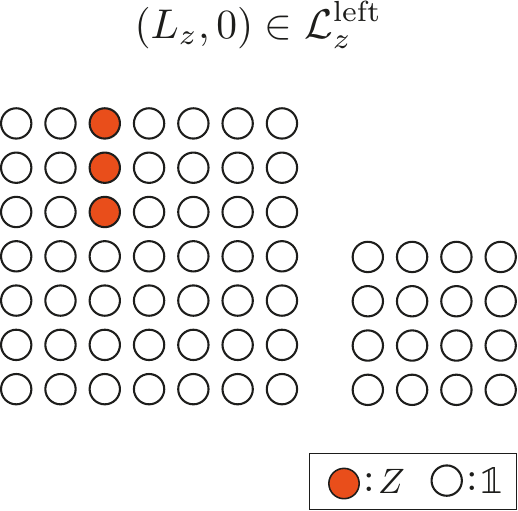}
\label{fig:zleft}
}\hfill
\subfloat[A logical right $Z$ operator.]{
\includegraphics[scale=0.45]{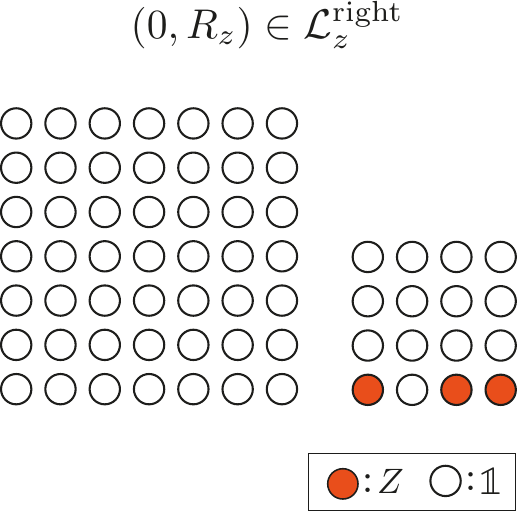}
\label{fig:zright}
}
\caption{A graphical representation of the qubit space of the homological product code $\cab$ where $\delta_A = \delta_B = \delta$ and $\delta$ is a degenerate parity check matrix for the $[7, 4, 3]$ Hamming code reported in (\ref{fig:delta}).
In (\ref{fig:xstab}), \dots, (\ref{fig:zright}), the two grids represent the left and right qubits respectively. One circle is drawn for every physical qubits of $\cab$. There are $7 \times 7$ left qubits and $4 \times 4$ right qubits. The support of an operator $(L, R)$ on $\cab$ is represented by filling the corresponding circle: left qubit at position $(j_a, j_b)$ is filled if and only if $L_{{j_a, j_b}} = 1$; similarly the right qubit at position $(i_a, i_b)$ is filled if and only if $R_{{i_a}, {i_b}} = 1$. The code $\cab$ pictured has parameters $[[65, 17, 3, 3]]$. It has $16$ independent logical left operators and $1$ logical right operators: $|\mathcal{L}_x^{\mathrm{left}}| = |\mathcal{L}_z^{\mathrm{left}}| = 16$ and $|\mathcal{L}_x^{\mathrm{right}}|=|\mathcal{L}_z^{\mathrm{right}}|=1$. 
}
\label{fig:qubits_grid}
\end{figure*}

\section{Results}
\label{sec:results}

Here we present the ReShape decoder. The intuition behind ReShape is that we can look at hypergaph product codes as codes built combining (product) multiple copies of the same classical codes. As such, with due care, we can `decouple' these copies and retrieve the original classical seed codes. 

On a $[[n, k, d]]$ hypergraph product code $\cab$, ReShape works by splitting the decoding problem into $k$ smaller classical decoding problems which can be solved using classical decoding algorithms for the seed matrices. In order to identify the $k$ classical decoding problems, it applies a linear transformation, a change of basis, on the $n$ dimensional codespace of $\cab$, yielding a canonical form for error operators. This canonical form exposes two important features of the codespace: the first one is that logical operators of $\cab$ are naturally partitioned into two sets, of left and right operators; the second is that the weight of each logical operator directly depends on the weight of the classical codewords of the seed codes. By writing an operator in its canonical form, we can immediately assess to which of the two classes it belongs and, via classical decoding, to which logical operator it is closest. Hence, we successfully detect and correct errors. 

In this Section, we first proceed to study the algebraic invariants of the logical operators upon which the canonical form is defined. The correctness of ReShape, and so the proof of Theorem \ref{thm:reshape_success}, strongly relies on the existence of these invariants. We detail the Reshape algorithm in Section \ref{sec:reshape} and discuss its limitations in Section \ref{sec:stochastic}. All the proof of this Section are deferred to Appendix \ref{app:proofs}.
 
\subsection{Invariants}
\label{sec:algebraic_properties}
The characteristic shape of operators on the codespace of $\cab$ and the structure of its stabilisers and logical operators, induces a canonical form for $Z$-operators in $\cab$. More precisely, by combining the construction outlined in Section \ref{sec:main} and the definition of complement of a vector subspace (see Appendix \ref{app:linear_algebra}) we have proven the following:
% constitutes give the following definition 
\begin{proposition}[Canonical form]
\label{prop:canonical_form}
Let $(L, R)$  be a $Z$-operator on the codespace of $\cab$. For a vector space $V \subseteq \ff^n$, we denote by $V^{\bullet}$ any space such that $V \oplus V^{\bullet} \simeq \ff^n$, (see Appendix \ref{app:linear_algebra}). Then, for the operator $(L, R)$, the left part $L$ can be expressed as a sum of a \emph{free part} $M_L$ and a \emph{logical part} $O_L$ such that every row of $M_L$ belongs to $\im \delta_B^T$ and every row of $O_L$ belongs to $(\im \delta_B^T)^{\bullet}$. Similarly, the right part $R$ can be expressed as a sum of a free part $M_R$ and a logical part $O_R$ such that every column of $M_R$ belongs to $\im \delta_A$ and every column of $O_R$ belongs to $(\im \delta_A)^{\bullet}$. Hence, for $(L, R)$ holds:
\begin{align}
\label{eq:cf}
\tag{CF}
(L, R) = (M_L + O_L, M_R + O_R).
\end{align}
We refer to the writing given by Eq.~\eqref{eq:cf} as \emph{canonical form} of the operator $(L,R)$.
\end{proposition}
Crucially, as we detail in Appendix \ref{app:proofs}, it is always possible to `move' the support of the free part of an operator from the left qubits to the right qubits and vice versa, by adding stabilisers. Opposite is the situation for the logical part: the support of the logical part of an operator cannot be moved from the left to the right qubits without changing its homology class. These two observations justify the name free and logical part in the canonical form of a $Z$-operator on $\cab$. We refer to Figure \ref{fig:qubits_grid} and \ref{fig:a_z_operator} for a visual representation of the canonical form of a $Z$-operator on $\cab$. In Figures \ref{fig:xstab} and \ref{fig:zstab} we see stabiliser operators in their canonical form: their free part has support pictured, their logical part is $0$. In Figures \ref{fig:xleft}, \ref{fig:xright}, \ref{fig:zleft}, \ref{fig:zright} we see logical operators in their canonical form: their free part is $0$, whilst their logical part, pictured, has support contained in either a line or a column of one of the two grids of qubits. In Figure \ref{fig:a_z_operator} we see a $Z$-operator whose free and logical part are both non trivial.

Given a $Z$-operator $(L, R)$ we define its \emph{row-column weight} as
\begin{defin}\label{definition:cr_weight}
Let $(L,\,R)$ be any $Z$-operator on the physical qubits of the code $\cab$. Its row-column weight is the integer pair:
\begin{equation*}
\wtrc(L,\,R) \coloneqq (\#\row(L),\, \#\col(R))
\end{equation*}
where
\begin{align*}
    \row(L) &\coloneqq \{L_i \text{ row of }L \,: L_i \neq 0\},\\
    \col(R) &\coloneqq \{R^j \text{ column of }R \,: R^j \neq 0\},
\end{align*}
and the hash symbol $\#$ denotes the cardinality of a set.  
\end{defin}

The primary significance of this novel notion of weight is explained by Proposition \ref{prop:logical-op}, which also represents a key result towards the construction of the ReShape decoder.

\begin{proposition}\label{prop:logical-op}
If $(L,\,R)$ is a non-trivial logical $Z$-operator of $\cab$ then either $\#\row(L) \ge d_a$ or $\#\col(R)\ge d_b^T$ (or both). 
\end{proposition}

Corollary \ref{cor:logical} below further specifies the structure of logical $Z$-operators and it is easily derived from the proof of Proposition \ref{prop:logical-op}, which is deferred to Appendix \ref{app:proofs}.

\begin{corollary}
\label{cor:logical}
If $(L, R)$ is a non-trivial logical $Z$-operator on $\cab$, at least one of the following hold:
\begin{enumerate}[label=(\roman*)]
\item $L$ has at least $d_a$ rows which are not in $\im \delta^T_B$ when seen as vectors in $\ff^{n_b}$.
\item $R$ has at least $d_b^T$ columns which are not in $\im \delta_A$ when seen as vectors of $\ff^{m_a}$. 
\end{enumerate}
\end{corollary}

Proposition \ref{prop:logical-op} and Corollary \ref{cor:logical} naturally yield 
\begin{defin}\label{definition:logical_weight}
Let $(L, R)$ be a $Z$-operator on $\cab$. Its logical row-column weight is the integer pair 
\begin{align*}
\wtrcl(L, R) &:= (\#\rowl(L), \#\coll(R))
\end{align*}
where
\begin{align*}
\rowl(L)& := \{L_i \text{ row of }L \,: L_i \not \in \im \delta_B^T\},\\
\coll (R)& := \{R^j \text{ column of }R \,: R^j \not \in \im \delta_A \}.
\end{align*}
Equivalently, if $(L, R)$ has canonical form given by
\begin{equation*}
  (L, R)=(O_L + M_L, O_R + M_R),
\end{equation*}
for its logical row-column weight holds:
\begin{equation*}
    \wtrcl(L, R) = \wtrc(O_L, O_R).
\end{equation*}
\end{defin}

\begin{figure}
\centering
\subfloat[A $Z$ operator.]{\includegraphics[scale =0.45]{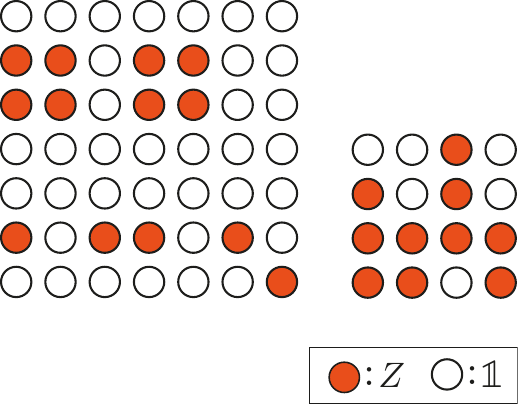}
\label{fig:entire_op}}\hfill
\subfloat[Its free part.]{\includegraphics[scale=0.45]{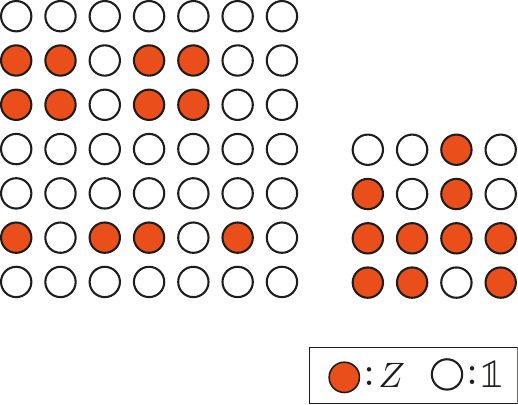}
\label{fig:free_part}
}\hfill
\subfloat[Its logical part.]{\includegraphics[scale=0.45]{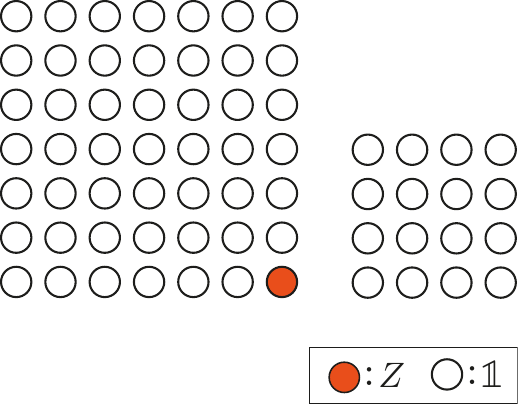}
\label{fig:logical_part}}
\caption{A graphical representation of a $Z$ operator on the physical qubits of the code $\cab$ also represented in Figure \ref{fig:qubits_grid}. The filled circles in (\ref{fig:entire_op}) represent the support of the operator: its Hamming weight is $23$ while its row-column weight is $(4, 4)$, see Definition \ref{definition:cr_weight}. The operator in (\ref{fig:entire_op}) can be written as a sum of its \emph{free} and its \emph{logical} parts represented in (\ref{fig:free_part}) and (\ref{fig:logical_part}), see Proposition (\ref{prop:canonical_form}). As per Definition \ref{definition:logical_weight} and figure (\ref{fig:logical_part}), the logical row-column weight of the operator in (\ref{fig:entire_op}) is $(1, 0)$. }
\label{fig:a_z_operator}
\end{figure}

The pivotal property of the logical row-column weight is expressed by Proposition \ref{prop:invariant}.
 
\begin{proposition}
\label{prop:invariant}
The logical row-column weight of a $Z$-operator on $\cab$ is an invariant of its homology class.
\end{proposition}
Proposition \ref{prop:invariant} not only justifies the introduction of the notion of logical row-column weight but also constitutes the core resource upon which we prove the correctness of the ReShape decoder, which we now introduce.

\subsection{The ReShape decoder} 
\label{sec:reshape}
An hypergraph product code $\cab$ is a CSS code and as such the decoding for $X$ and $Z$ error can be treated separately but in a symmetric way. Here we focus on $Z$-errors and therefore we measure a generating set of $X$-stabilisers.  The $Z$-error decoding problem for $\cab$ can be stated as: given a $m_a \times n_b$ syndrome matrix $S$, find a valid and correct solution $(\tilde L, \tilde R)$ to the equation:
\begin{align}\label{eq:se}
S=\sigma(L,\,R)\coloneqq \delta_A L + R \delta_B, \tag{SE}
\end{align}
where $(\tilde L, \tilde R)$ is \emph{valid} if $\sigma(\tilde L, \tilde R) = S$ and it is \emph{correct} if it belongs to the homology class of the minimum weight operator with syndrome $S$. Crucially, finding \emph{a valid} solution $(L, R)$ to Eq.~\eqref{eq:1_se} is always possible by solving the linear system of equation where the parity check matrix of $X$ stabilisers is the matrix of coefficients and the syndrome $S$ is the constant term. The difficulties arise if we are interest in finding \emph{a correct} solution to Eq.~\eqref{eq:1_se}.

The ReShape decoder for $Z$-errors is build upon two classical minimum weight decoding algorithms: $\mathscr D_{\delta_A}$ and $\mathscr D_{\delta_B^T}$. By this we mean that (i) the algorithms $\mathscr D_{\delta_A}$ and $\mathscr D_{\delta_B^T}$ are optimal decoders for the classical linear code with parity check matrix $\delta_A$ and $\delta_B^T$ respectively, and (ii) they solve the classical decoding problem of Eq.~\eqref{eq:classical_decoding} for errors of weight up to $(d_a-1)/2$ and $(d_b^T-1)/2$ respectively. 
Reshape takes as input $\mathscr{D}_{\delta_A}$, $\mathscr{D}_{\delta_B^T}$, a syndrome matrix $S$ and a valid solution $(L,\,R)$ of the Syndrome Equation \eqref{eq:se}: $\sigma(L, R) = S$. Recall that a valid solution $(L, R)$ can always be efficiently found either solving the associated linear system or querying a lookup table. It outputs a correct solution of \eqref{eq:se}: an operator $(\tilde L, \tilde R)$ homologically equivalent to the minimum weight operator $(L_{\min}, R_{\min})$ with syndrome $S$.

ReShape (Algorithm \ref{algo:reshape}) works separately on the left part $L$ and on the right part $R$ of the operator $(L,\,R)$ and in fact it could be run in parallel (lines \ref{alg:for1_left}-\ref{alg:for2_left_end} and lines \ref{alg:for1_right}-\ref{alg:for2_right_end}). Starting from a valid solution $(L,\,R)$, it minimises its logical row-column weight by minimizing $\#\rowl(L)$ (lines \ref{alg:for1_left}-\ref{alg:for2_left_end}) first and $\#\coll(R)$ after (lines \ref{alg:for1_right}-\ref{alg:for2_right_end}). Because the logical row-column weight is an homology invariant for $Z$-operators (Proposition \ref{prop:invariant}) and ReShape minimises it, this suffices to assure that ReShape is correct, as stated in Proposition \ref{prop:decoder}.
\begin{algorithm}
\caption{ReShape decoder for $Z$-errors.}
\begin{algorithmic}[1]
\Require Classical decoder $\mathscr{D}_{\delta_A}$ and $\mathscr{D}_{\delta_B^T}$. Syndrome matrix $S$ and operator $(L,\,R)$ on $\cab$ s.t. $\sigma(L,\,R) = S$.
\newline
\Ensure Operator $(\tilde L, \tilde R)$ on $\cab$ s.t. $\sigma(\tilde L, \tilde R) = S$ and $[\tilde L, \tilde R] = [L_{\min}, R_{\min}]$, where $(L_{\min}, R_{\min})$ is a minimum weight operator with syndrome $S$.
\newline
\ForAll{$L_i$ rows of $L$}\label{alg:for1_left}
    \State{$\mathrm{Split}$: $L_i = m_i+ \mu_i \in \im \delta_B^T \oplus (\im \delta_B^T)^\bullet$, \quad as in~\eqref{eq:nice_basis}}
\EndFor\label{alg:for1_left_end}
\State{$M_L \gets$ matrix whose rows are $m_i$}
\State{$O_L \gets$ matrix whose rows are $\mu_i$}
\ForAll{$O_L^j$ columns of $O_L$}\label{alg:for2_left}
\State{$\mathrm{Decode}$: $\rho^j = \mathscr{D}_{\delta_A}(O_L^j)$}\label{alg:decode}
\EndFor\label{alg:for2_left_endfor}
\State{$\tilde L \gets$ matrix whose columns are $\rho^j$} \label{alg:create_l}
\State{$\tilde L \gets \tilde L + O_L + M_L$}\label{alg:for2_left_end}
\ForAll{$R^j$ columns of $R$} \label{alg:for1_right}
\State{$\mathrm{Split}$: $R^j = m^j+ \mu^j \in \im \delta_A \oplus (\im \delta_A)^\bullet $}, \quad as in \eqref{eq:nice_basis}
\EndFor
\State{$M_R \gets$ matrix whose columns are $m^j$}
\State{$O_R \gets$ matrix whose columns are $\mu^j$}
\ForAll{${(O_{R})}_i$ rows of $O_R$}\label{alg:for2_right}
\State{$\mathrm{Decode}$: $\rho_i = \mathscr{D}_{\delta_B^T}({(O_{R})}_i)$}
\EndFor
\State{$\tilde R \gets$ matrix whose rows are $\rho_i$}
\State{$\tilde R \gets \tilde R + R_L + M_R$}\label{alg:for2_right_end}
\State\Return{$(\tilde L, \tilde R)$}
\end{algorithmic}
\label{algo:reshape}
\end{algorithm}
ReShape works on the left part $L$ of the inputted valid solution $(L, R)$ (lines \ref{alg:for1_left}-\ref{alg:for2_left_end}) into two steps: Decode and Split. Each of these two steps exploits a characteristic feature of the $Z$-operators on the codespace of $\cab$: 
\begin{enumerate}[label=(\roman*)]
    \item Split step: a $Z$-stabilizer $(G_L, G_R)$ has left part $G_L$ such that every row is in the image of $\delta_B^T$;
    \item Decode step: a logical $Z$-operator which acts non-trivially on the left qubits has a representative $(L_z, R_z)$ such that at least one column of $L_z$ is in $\ker \delta_A\setminus\{0\}$.
\end{enumerate}
The Split and Decode steps are similarly performed on the right part $R$, as specified in lines \ref{alg:for1_right} - \ref{alg:for2_right_end} of the pseudocode in Algorithm \ref{algo:reshape}. Again with reference to the left part as guide case, we now describe the Split and Decode steps in details and specify their computational cost. By extending this analysis to the right part, and thanks to Proposition \ref{prop:decoder}, Theorem \ref{thm:reshape_success} is proved. 

Let $(L,\,R)$ be any valid solution of \eqref{eq:se} given in input to ReShape.

(i) Split. First, in lines \ref{alg:for1_left}-\ref{alg:for1_left_end}, $L$ is written in its canonical form with respect to the basis described by Eq.~\eqref{eq:nice_basis}:
\begin{align*}
    L = M_L + O_L.
\end{align*}
This operation has the cost of a change of basis over the vector space $\ff^{n_b}$, namely from the canonical basis to the basis described by Eq.~\eqref{eq:nice_basis}. A change of basis over a vector space is a linear operation that correspond to a multiplication by an invertible square matrix. Since we are interested in computing the image of this linear transformation for each of the $n_a$ column vectors of $L$, this amount to the multiplication of an $n_a \times n_b$ and a $n_b \times n_b$ matrix. To sum up, the Split step of ReShape has cost $O(n_a n_b^2)$.

(ii) Decode. The second step performed by ReShape (lines \ref{alg:for2_left}-\ref{alg:for2_left_end}) aims to minimise the logical row-column weight of $(L,\,R)$ by looking at non-homologically equivalent operators:
\begin{align*}
    (L,\,R) + (L_z, 0),
\end{align*}
as $L_z$ varies in $\mathcal L_z^{\mathrm{left}}$. More precisely, ReShape exploits the canonical form of $L$ computed at the previous step and scans through all the columns of its logical part $O_L$. By construction, any row $(O_L)_i$ of $O_L$ belongs to the complement of $\im \delta_B^T$. For this reason, $(O_L)_i$ can be written as the linear combination of $k_b = n_b - \rk(\delta_B)$ unit vectors in $\ff^{n_b}$ which does not belong to $\im \delta_B^T$ i.e.\ $k_b$ unit vectors which span $(\im \delta_B^T)^{\bullet}$, see Appendix \ref{app:linear_algebra}. Importantly, when we stack these row vectors $(O_L)_i$ all together and consider the columns of the matrix $O_L$, we observe that $O_L$ cannot have more than $k_b$ non-zero columns. Each non-zero column of $O_L$ is then treated as if it corresponded to a code-word of the classical code with parity check matrix $\delta_A$ (plus eventually a noise vector) and decoded individually using the classical minimum weight decoder $\mathscr D_{\delta_A}$ (line \ref{alg:decode}):
\begin{align}
\label{eq:decode_min}
    \mathscr{D}_{\delta_A}(O_L^j) = \rho^j \Longleftrightarrow \rho^j \ \in \arg\min_{k \in \ker \delta_A }(|k + O_L^j|).
\end{align}
If the computational cost of the classical decoder $\mathscr{D}_{\delta_A}$ is $O(c_a)$, the computational cost of the second step of ReShape is $O(k_b c_a)$.

The Split and Decode steps described for the left part are replicated, with opportune modifications, for the right part. To be exact, if one or both $\delta_A$ and $\delta_B$ are full rank, then the right part does not encode any logical operator so the algorithm terminates\footnote{In fact, as per Eq.~(\ref{eq:matrix_logical_z}) and Eq.~ (\ref{eq:matrix_logical_x}), if $\rk(\delta_A) = m_a$ or $\rk(\delta_B) = m_b$, then $\ker \delta_A^T = (\im \delta_A)^{\bullet} = \{0\}$ or $\ker \delta_B^T = (\im \delta_B)^{\bullet} = \{0\}$ and so $\mathcal{L}^{\mathrm{right}}_x$ = $\mathcal{L}^{\mathrm{right}}_z$.}.

Proposition \ref{prop:decoder} below ensures that the recovery operator $(\tilde L, \tilde R)$ found by ReShape is a correct solution of \eqref{eq:se}, as long as the classical decoders $\mathscr D_{\delta_A}$ and $\mathscr{D}_{\delta_B^T}$ succeed. 
\begin{proposition}\label{prop:decoder}
Let $S$ be an $X$-syndrome matrix for $\cab$ and $(L, R)$ any valid solution to the Syndrome Equation:
\begin{align*}
    \tag{\ref{eq:se}}\sigma(L,\,R)=S.
\end{align*}
Suppose that the minimum weight operator $(L_{\min},\, R_{\min})$ with syndrome $S$ has $(d_a/2, d_b^T/2)$-bounded logical row-column weight i.e.\:
\begin{align*}
    \wtrcl(L_{\min},\, R_{\min}) =(\#\rowl(L_{\min}), \#\coll(R_{\min})),
\end{align*}
is such that
\begin{align}
\label{eq:teo_weight}
    \#\rowl(L_{\min}) < \frac{d_a}{2} \quad\text{ and }\quad
    \#\coll(R_{\min}) < \frac{d_b^T}{2}.
\end{align}
Then, on input $\mathscr{D}_{\delta_A}$, $\mathscr{D}_{\delta_B^T}$, $S$ and $(L,\,R)$, ReShape outputs a correct solution $(\tilde L,\, \tilde R)$ of \eqref{eq:se}, provided that the classical decoders $\mathscr{D}_{\delta_A}$, $\mathscr{D}_{\delta_B^T}$ succeed. In other words, the solution $(\tilde L, \tilde R)$ found by ReShape is in the same homology class as the minimum weight operator with syndrome $S$: 
\begin{align*}
    [L_{\min}, \, R_{\min}] = [\tilde L,\, \tilde R].
\end{align*}
\end{proposition}
It is important to note that the condition \eqref{eq:teo_weight} on the weight of the original error is on its row-column weight, while usually decoding success is assessed depending on the weight of an operator, meaning the number of its non-identity factors. Obviously, for any operator $(L,\,R)$ it holds:
\begin{equation*}
\#row(L) \le |L| \quad\text{and}\quad \#col(R) \le |R|.
\end{equation*}
As a consequence, Proposition \ref{prop:decoder} entails that ReShape succeeds in correcting any $Z$-error of weight up to half the code distance $d_z = \min\{d_A, d_B^T\}$. Combining this with the cost analysis of the Split and Decode steps detailed above, gives a proof of Theorem \ref{thm:reshape_success}.

It is worth to observe that actually ReShape can correct errors of weight strictly bigger than half the code distance, as long as they are not too `spread'. In fact, whenever an error is homologically equivalent to an operator $(L, R)$ such that $L$ has `few' non-zero rows and $R$ has `few' non-zero columns, ReShape succeeds. Formally, because by definition:
\begin{equation*}
\#\row(L)\ge\#\rowl(L) \quad\text{and}\quad \#\col(R) \ge \#\coll(R).
\end{equation*}
Proposition \ref{prop:decoder} yields
\begin{corollary}
Provided that the classical decoders succeed, ReShape successfully corrects any $Z$-error $(L, R)$ with bounded row-column weight:
\begin{align*}
      \#\row(L) < \frac{d_a}{2} \quad\text{ and }\quad
    \#\col(R) < \frac{d_b^T}{2}.  
\end{align*}
\end{corollary}
To sum up, ReShape successfully solves the decoding problem for any hypergraph product code requiring only $k$ oracle calls to a classical decoder for the seed matrices, where $k$ is the logical dimension of the code. Furthermore, it is able to correct for a vast class of errors of weight strictly bigger than half the code distance, provided that they have a `good' shape. Here by `good' we mean errors of low logical column-row weight but arbitrary Hamming weight as for instance the $Z$-operator pictured in Figure \ref{fig:a_z_operator}, that has Hamming weight $23$ but logical row-column weight $(1, 0)$ and would therefore be successfully corrected by the ReShape decoder.

The next Section focuses on what happens when we cannot control the shape of the errors but we assume that the probability of a given error to occur decays exponentially in its weight.
\subsection{ReShape for Stochastic noise}
\label{sec:stochastic}
Up till now, we have focused on the adversarial noise model: errors on qubits are always correctable because we assume they have weight less than half the code distance. In real systems though, this is rarely the case and it is more faithful to assume that errors are sampled accordingly to a local stochastic noise model, where qubits errors have arbitrary location but the probability of a given error decays exponentially in its weight \cite{dennis02}. More precisely the probability of a Pauli error $E \in \mathcal P_n$ to occur is given by:
\begin{equation}
\label{eq:prob_error}
    \mathbb P(E) = p^{|E|} (1-p)^{n-|E|},
\end{equation}
meaning that Pauli errors on each of the $n$ qubits are independent and identically distributed.
Under the binomial distribution associated to Eq.~\eqref{eq:prob_error}, the expected error weight on the encoded state is $pn$. Because the best possible distance scaling for the hypergraph product codes is $\sim \sqrt{n}$ (when the classical seed codes have linear distance), as $n$ increases, we eventually find $pn > \sqrt{n}/2 \sim d/2$. Nonetheless, it is well known that LDPC hypergraph product codes do have a positive error correcting threshold \cite{Kovalevbadcode}. A family of codes has threshold $p_{\mathrm{th}} >0$ if, for noise rate below $p_{\mathrm{th}}$, non-correctable errors that destroy the logical information occur with probability $p_{\mathrm{non-correctable}}$ which decays exponentially in the system size \cite{Dennis01, Kovalevbadcode, fawzi2018}:
\begin{align}
\label{eq:threshold}
   p_{\mathrm{non-correctable}}\propto\left( \frac{p}{p_{th}}\right)^{\alpha d^{\beta}}
\end{align}
for some $\alpha, \beta >0$. It is important to stress that Eq.~\eqref{eq:threshold} does not contrast with the fact that the typical error under the stochastic noise model will have weight $pn$. Instead, Eq.~\eqref{eq:threshold} entails that, among all the errors sampled, the non-correctable ones are only a small fraction. Beyond the theoretical threshold that Kovalev and Pryadko proved in \cite{Kovalevbadcode}, the literature offers several numerical evidence of decoders for hypergraph product or related families of codes which exhibit a threshold. Nonetheless these decoders either lack a correctness proof, e.g. BP in \cite{panteleev2019degenerate, roffe2020decoding}, or need some additional constraints on the seed matrices, e.g. expander codes with small-set flip decoder \cite{fawzi2018}, or augmented surface codes with the union-find decoder in \cite{delfosse2020union}.

On the contrary, for any choice of the seed matrices in the hypergraph product, ReShape is provably correct for adversarial errors. Not surprisingly though, ReShape does not show a threshold and a possible intuition for its anti-threshold behaviour is the following. 

If we contrast Reshape with pairs of LDPC codes families and decoders which exhibit a threshold, such as expander codes with the small-set flip decoder \cite{fawzi2018} or hypergraph product codes with BP \cite{ panteleev2019degenerate, roffe2020decoding}, a feature of difference is the `locality' of the decoding algorithms. % Here locality is intended of the decoding algorithm with respect to the qubit graph, where nodes are qubits that are connected if and only if there exist a stabiliser generator whose support contains both. 
Loosely, we say that a decoding algorithm is local if errors affecting distant regions on the qubit graph are dealt with separately and independently. We stress that a decoder's locality is a feature of the algorithm and it is not related to the locality of the code's stabiliser generators. A code can have local stabilisers, meaning that for a given layout of qubits in the space, stabiliser generators only involve qubits in a limited area, and yet be equipped with a non-local decoding algorithm. Indeed, ReShape is such a decoder.  It is a non-local decoder that can be used on the very much local planar code. Locality of the decoding algorithm is relevant because local stochastic errors tend to form small disjoint clusters on the qubit graph which do not destroy the logical information as long as they are (1) small enough (2) sufficiently far apart. Therefore, if a decoder manages to mimic the error cluster distribution on the qubit graph and finds recovery operators accordingly, then it is likely to preserve the logical information and show a threshold. ReShape, on the other hand, has a deeply global nature. The Split step \emph{groups} all the clusters of flipped qubits scattered across the qubit graph in a small pre-assigned region; a recovery is then chosen (Decode step) based on the syndrome information in this pre-assigned region. If we take the planar code as an example (see Figure \ref{fig:reshape_example}), the Split step groups the error (and the syndrome) weight on one column of the left qubits. The subsequent Decode step decodes that column and finds a recovery operator with supported on the column. Because for the planar code a logical $Z$-operator can be chosen to have support on only one column of the left qubits, this procedure can easily destroy the logical information.

% A common ingredient for decoders that exhibit a threshold on LDPC codes is their local nature on the syndrome graph, where nodes are stabiliser generators that are connected by an edge if and only if their support overlap. Roughly, local stochastic errors tend to form small disjoint clusters on the qubit graph\footnote{i.e.\ nodes are qubit and two qubits are connected if and only if there exists a stabiliser generator whose support contains both qubits.} which do not destroy the logical information as long as they are (1) small enough (2) sufficiently far apart. If a decoder manages to mimic the error cluster distribution on the syndrome graph and finds recovery operators accordingly, then it is likely to show a threshold. ReShape, on the other hand, has a deeply global nature. The Split step \emph{groups} all the clusters of flipped stabilisers scattered across the syndrome graph in a small pre-assigned region; a recovery is then chosen (Decode step) based on the syndrome information there  contained. If we take the planar code as an example (see Figure \ref{fig:reshape_example}), the Split step groups the error (and the syndrome) weight on one column of the left qubits. The subsequent Decode step decodes that column and find a recovery operator with support there contained. Because for the planar code a logical $Z$-operator can be chosen to have support on only one column of the left qubits, this procedure can easily destroy the logical information. 

\begin{figure*}[t]
    \centering
    \includegraphics[width=\linewidth, keepaspectratio]{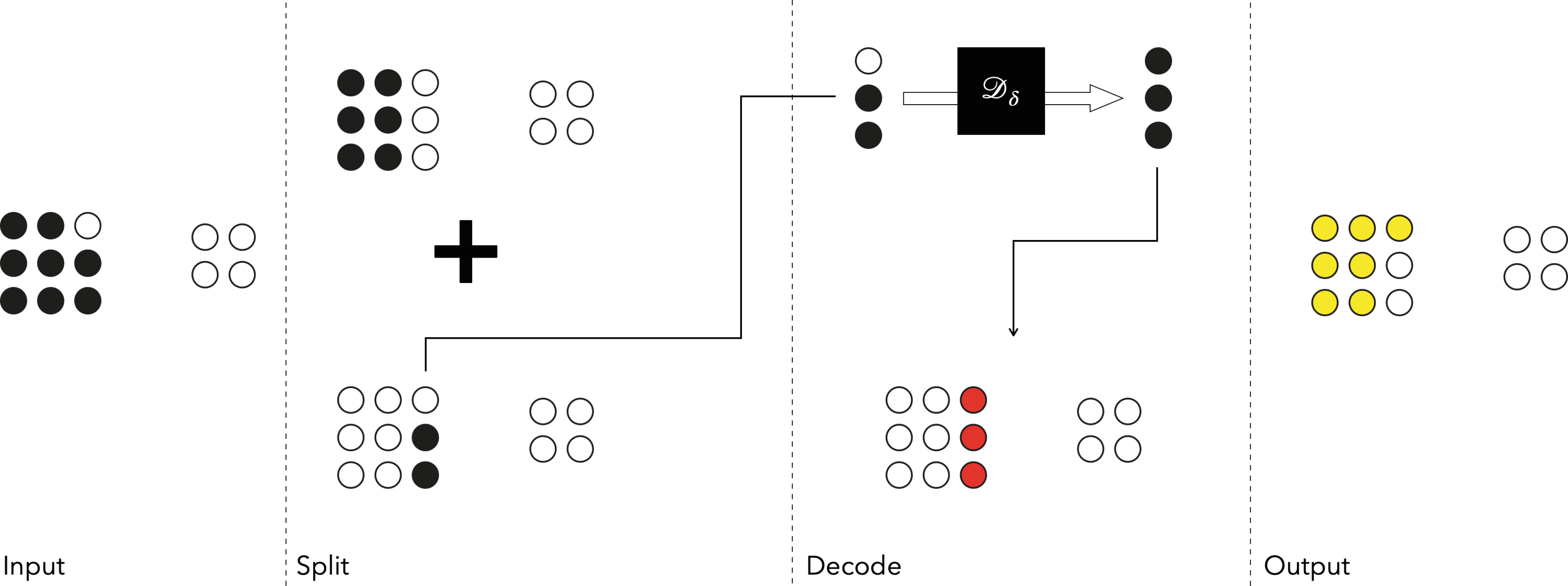}
    \caption{Graphical representation of one instance of ReShape for $Z$-errors. The code considered is the planar code of distance $3$ (toric code with boundaries) or, equivalently, the $[[13, 1, 3]]$ hypergraph product code $\mathcal{C}(\delta, \delta)$ for $\delta$ full-rank parity check matrix of the distance-$3$ repetition code i.e.\,leftmost matrix in (\ref{eq:repetition_code}). We use the same graphical representation used in Figure \ref{fig:qubits_grid}. A minimum weight logical $Z$ operator for $\mathcal{C}(\delta, \delta)$ can be chosen to have support on all the qubits of a column of left qubits (Decode, bottom grid of qubits, support in red). The row span of $\delta$ consists of all vectors in $\ff^3$ of even weight, hence for a generic $Z$-operator on $\mathcal{C}(\delta, \delta)$, all the rows of left qubits that have an even number of filled qubits belongs to its free part and do not contribute to its logical row-column weight. Since $\delta$ is full-rank, its column span is the whole space $\ff^2$, and therefore a column displaying any choice of filled qubits is in the image of $\delta^T$. As such, there is no contribution to the logical-row column weight from the right part of the operator. In particular, there is no need to run the ReShape decoder on the right part: Algorithm \ref{algo:reshape} will not execute lines \ref{alg:for1_right} - \ref{alg:for2_right_end}. \protect\newline
  The figure is divided into to four sectors, one for each stage of the decoding cycle: Input, Split, Decode and Output. \protect\newline
    \textbf{Input}: the valid solution in input $(L, R)$ has support represented by the black qubits. The operator $(L, R)$ has Hamming weight $8$.\protect\newline
    \textbf{Split} - Algorithm \ref{algo:reshape}, lines \ref{alg:for1_left} - \ref{alg:for1_left_end} : $(L, R)$ is written as sum of its free part $(M_L, M_R)$ (top); and its logical part $(O_L, O_R)$ (bottom). For what said on the image of $\delta$, the free part $M_L$ of $L$ is a matrix whose rows have all even Hamming weight. Since $O_L$ has $2$ non-zero rows, $(L, R)$ has logical row-column weight $\wtrcl(L, R) = 2$. \protect\newline
    \textbf{Decode} - Algorithm \ref{algo:reshape}, lines \ref{alg:for1_left} - \ref{alg:create_l} : the non-zero column $(0, 1, 1)^T$ of the logical part $O_L$ of $L$ is given in input to a decoder $\mathscr{D}_{\delta}$ for the classical distance-3 repetition code. The solution found is $(1, 1, 1)^T$, represented by the single column of black bits on the top. This solution is plugged in the hypergraph product code and yields a logical operator correction represented by the operator at the bottom with support on the red qubits.\protect\newline
    \textbf{Output} - Algorithm \ref{algo:reshape}, line \ref{alg:for2_left_end} : the output solution $(\tilde L, \tilde R)$ is obtained by adding the input operator $(L, R)$ and the operator found in the Decode step with support on the red qubits. The support of $(\tilde L, \tilde R)$ is represented by the yellow qubits. \protect\newline
    We note that the solution found $(\tilde L, \tilde R)$ has Hamming weight $7$ and, by observing that only the first rows has odd weight, we deduce that its logical row-column weight is $1$. It is easy to verify that $(\tilde L, \tilde R)$ is indeed homologically equivalent to the minimum weight solution $(\hat L_{1,3}, 0)$ where $\hat L_{1, 3} \in \ff^{3 \times 3}$ is the matrix with all zeros but for the $(1, 3)$-th entry which is $1$. In fact, $(\tilde L, \tilde R) = (\hat L_{1, 3}, 0) + S^z(1, 1) + S^z(2, 1) + S^z(3, 1)$ thus, by \eqref{eq:homology}, $[\tilde L, \tilde R] = [\hat L_{1, 3}, 0]$.}
    \label{fig:reshape_example}
\end{figure*}
Our intuition on the performance of ReShape under stochastic noise finds confirmation in the plots reported in Figure \ref{fig:anti_threshold}. Even if at first sight the plots in Figure \ref{fig:anti_threshold_toric} and \ref{fig:anti_threshold_ldpc} could indicate the presence of a very low threshold (below $1 \%$), a closer analysis suggests that this is not the case. In fact, as $d$ increases, the crossing point between the dashed curve labelled $d=0$ and the $d$-curves slips leftwards. Since the dashed curve is the locus of points where the failure probability $p_{\mathrm {fail}}$ equates the noise rate $p$, it corresponds to the case of no encoding i.e.\ $d=0$. The common crossing point, in other words, represents the \emph{pseudo-threshold} of the code \cite{svore2005flow}. Importantly, if a code family has a threshold $p_{\mathrm{th}}$ in the sense of Eq.~\eqref{eq:threshold}, then all the codes of the family crosses the curve $d=0$ at the same point of coordinates $(p_{\mathrm{th}},p_{\mathrm{th}} )$. 
Figure \ref{fig:pseudo_toric} clearly illustrates this left slipping phenomenon for the toric codes without boundaries. For close distances $d=6,8,10$, there it seems to be a common crossing point with the $d=0$ curve. However, the crossing point lowers if we increase $d$ more substantially, e.g. $d=20$. The situation appears less clear in Figure \ref{fig:pseudo_ldpc} because the pseudo-threshold seems to increase with the distance of the code. Still though, there is no common crossing point of the three curves; besides, we would expect the same trend as the one observed for the toric codes if codes of bigger distance were considered.

In conclusion, ReShape is not suited to tackle stochastic errors on $[[n, k, d]]$ code in the regime where typical errors have weight exceeding $d/2$. 

\begin{figure*}[tb]
\centering
    \subfloat[]{   
    \label{fig:anti_threshold_toric}
    \includegraphics[scale=0.85]{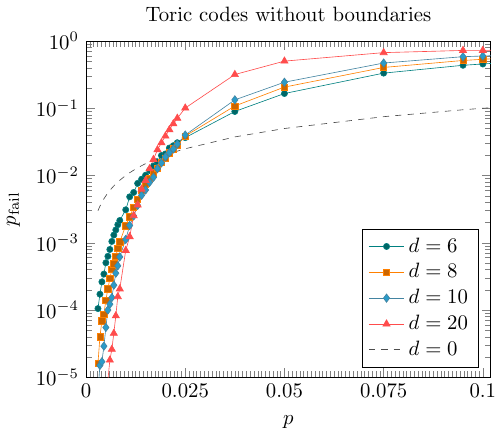}}\hfill
     \subfloat[]{  
     \label{fig:anti_threshold_ldpc}
     \includegraphics[scale=0.85]{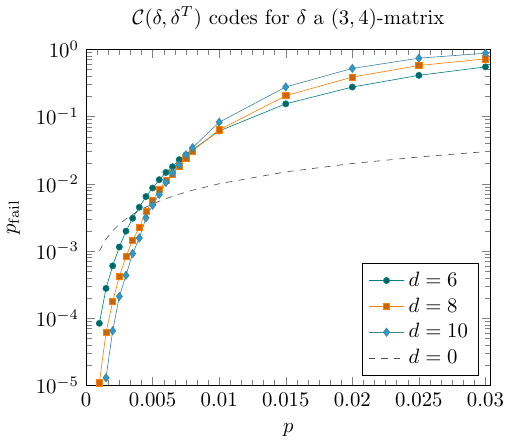}}
     \vfill
     \subfloat[]{\label{fig:pseudo_toric}\includegraphics[scale=0.85]{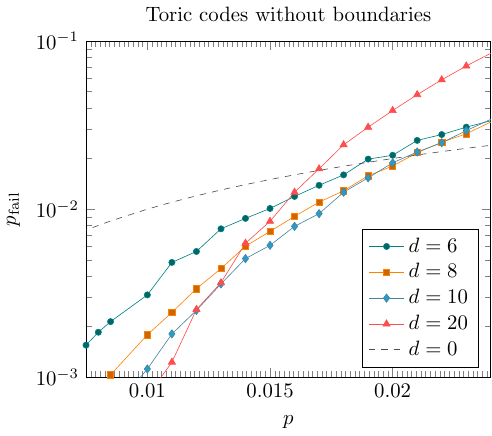}}
     %width=0.45\linewidth, keepaspectratio
     \hfill
     \subfloat[]{\label{fig:pseudo_ldpc}\includegraphics[scale=0.85]{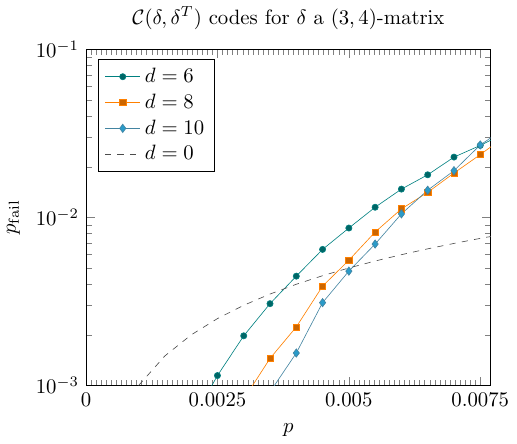}}
    \caption{Anti-threshold behaviour of ReShape on two families of hypergraph product codes. For both families, we plot the failure probability $p_{\mathrm{fail}}$ as a function of the phase-flip noise rate $p$ for codes with $Z$-distance $d$. All data points are generated with at least $10^4$ Monte Carlo trials. (a) and (c): The toric code without boundaries. (b) and (d): A family of hypergraph product codes $\mathcal{C}(\delta, \delta^T)$ where $\delta$ is a full-rank $(3, 4)$-matrix randomly generated, see \cite{roffe2020decoding}.}
    \label{fig:anti_threshold}
\end{figure*}

\section{Conclusions and outlook}
In this paper we determined some important homology invariants of hypergraph product codes. Exploiting these invariants, we designed the ReShape decoder. ReShape is the first decoder to efficiently decode for all errors up to half the code distance, across the whole spectrum of hypergraph product codes.

We foresee two natural extensions of this work.
The first is to adapt ReShape for it to work in the stochastic noise model settings. Because ReShape actually succeeds in correcting errors of weight substantially bigger than $(d-1)/2$ (namely it corrects error of weight as big as $\sim d^2$, when they have the right shape!), this gives us some hope that ReShape would work under stochastic noise if paired with the right clustering technique. For instance, something on the line of the clustering methods used in the renormalisation group or the union-find decoders \cite{bravyi2011analytic, bravyi2013quantum, delfosse2017almost, delfosse2021union_ldpc}. 

The second is to find the corresponding invariants for other families of homological product codes. Specifically, for the codes in \cite{evra2020decodable}, which have `rectangular' shaped logical operators instead of `string' like as the standard hypergraph product codes here studied; or the balanced product codes proposed in \cite{breuckmann2020balanced}. Once found, the right invariants could be plugged-in an appropriately modified version of ReShape and yield a provable correct decoder for these class of codes too.

\appendices
\section{Linear algebra: Space complement}
\label{app:linear_algebra}
In this Appendix we review some known linear algebra facts that we use in our proofs. We refer the reader for instance to \cite{lang2005undergraduate, herstein1996abstract} for a detailed presentation on the topic.

Consider a $m \times n$ binary matrix $\delta$. If $\rk(\delta) = \mathrm{rk}$ then we can choose binary vectors $v_1, \dots, v_{\mathrm{rk}}$  in $\mathbb F_2^{m}$ whose span is $\im \delta$:
\begin{align*}
\im \delta = \langle v_1, \dots, v_{\mathrm{rk}} \rangle.
\end{align*}
Let $\id$ be the $m \times m$ identity matrix. Perform Gaussian reduction on the $(\mathrm{rk} + m)$-row matrix $M$:
\begin{align*}
    M = \begin{pmatrix} v_1 \\
    \vdots\\
    v_{\mathrm{rk}}\\
    \hline
    \id
    \end{pmatrix}.
\end{align*}
By selecting the pivot rows, we obtain a basis of $\mathbb F_2^m$ of the form:
\begin{align}
\label{eq:nice_basis}
    \{v_1, \dots, v_{\mathrm{rk}}, f_{\mathrm{rk}+1}, \dots, f_{m}\},
\end{align}
where the $f_i$ are unit vectors. Letting:
\begin{equation*}
     (\im \delta)^{\bullet} \coloneqq \langle f_{\mathrm{rk}+1}, \dots, f_{m}\rangle,
\end{equation*}
we have:
\begin{equation}
\label{eq:splitting_spaces}
    \mathbb F^m = (\im \delta ) \oplus (\im \delta)^{\bullet}.
\end{equation}
We refer to the space $(\im \delta)^{\bullet}$ as complement of the space $\im \delta$.
We remark that the complement $V^{\bullet}$ is not equal to the orthogonal complement $V^{\bot}$. To see how this is the case, consider
\begin{align*}
    V = \langle \begin{pmatrix}
    1 \\ 1 \\1
    \end{pmatrix}, \begin{pmatrix}
    0 \\ 1 \\ 0
    \end{pmatrix} \rangle.
\end{align*}
Then the spaces $V^{\bullet}$ and $V^{\bot}$ can be chosen as
\begin{align*}
    V^{\bullet} = \langle \begin{pmatrix}
    0 \\ 0 \\ 1
    \end{pmatrix}\rangle , && V^{\bot} = \langle\begin{pmatrix}
    1 \\ 0 \\ 1
    \end{pmatrix}  \rangle.
\end{align*}
In particular, $V^{\bot} \subseteq V$ while $V^{\bullet} \cap V = \{0\}$.
\section{Hypergraph product codes}
\label{app:hpg_codes}
CSS codes can be easily described in terms of homology theory \cite{Kit03, bombin2007homological, bravyi2014homological} via the identification of the objects of the code with a chain complex \cite{hatcher2002algebraic}. For our purposes, a length $\ell$ chain complex is an object described by a sequence of $\ell + 1$ vector spaces $\{\mathcal C_i\}_i$ over $\mathbb F_2$ and $\ell$ binary matrices $\{\partial_i : \mathcal C_i \longrightarrow \mathcal C_{i+1}\}_i$ such that, for each $i$, $\partial_i \partial_{i-1} = 0$. In the following, we use the symbol $\partial$ to indicate the maps of a chain complex of length $\ell > 1$ and the symbol $\delta$ to indicate the map of a chain complex of length $1$. Given a chain complex $\mathfrak C$:
\begin{equation}
\label{eq:chain_css}
\tag{$\mathfrak C$}
\mathcal C_{-1} \xrightarrow[]{\partial_{-1}} \mathcal C_0 \xrightarrow[]{\partial_0}\mathcal C_1,
\end{equation}
we can define a CSS code $\mathcal{C}$ by equating:
\begin{align*}
    H_Z = \partial_{-1}^T,\quad H_X = \partial_0.
\end{align*}
Since $\partial_0\partial_{-1}=0$ by construction, $X$-type and $Z$-type operators do commute i.e.\ $H_X\cdot H_Z^T = 0$ and the code $\mathcal C$ associated to the chain complex \eqref{eq:chain_css} is well defined. The code $\mathcal C$ has length $n = \dim(\mathcal{C}_0)$ and its dimension $k$ equates to the dimension of the $0$th homology group $\mathcal{H}_0 = \ker\partial_0/\im \partial_{-1}$ or, equivalently, to the dimension of the $0$th co-homology group $\mathcal{H}_0^* = \ker{\partial_{-1}}/\im \partial_0$. Its $Z$-distance and $X$-distance are given by the minimum Hamming weight of any representative of a non-zero element in $\mathcal{H}_0$ and $\mathcal{H}_0^*$ respectively:
\begin{align*}
    d_z &= \min_{v \in \mathbb F_2^{n}} \{|v| : [v] \in \mathcal{H}_0, [v] \neq 0 \},\\
    d_x &= \min_{v \in \mathbb F_2^{n}} \{|v| : [v] \in \mathcal{H}_0^{*}, [v] \neq 0\}.
\end{align*}
An hypergraph product code $\cab$, which is a CSS code, can be easily defined in terms of product of chain complexes.
Consider the two length-1 chain complexes defined by the seed matrices $\delta_A$ and $\delta_B$:
\begin{align*}
    C^0_A &\xrightarrow[]{\delta_A} C^1_A,\\
   C^0_B &\xrightarrow[]{\delta_B} C^1_B.
\end{align*}
We define their homological product as follows. Take the tensor product spaces $\mathcal C_{-1}= C^0_A \ox C_B^1$ and $\mathcal C_1 =C^1_A \ox C^0_B$ and the direct sum space $\mathcal C_0 = C^0_A \ox C^0_B \oplus C^1_A \ox C^1_B$. The chain complex $\mathfrak{
C}_{A, B}$:
\begin{equation}
\label{eq:cab}
\tag{$\mathfrak{C}_{A, B}$}
\begin{tikzcd}[scale cd=0.75]
& C^1_A \ox C^0_B\arrow[from=2-1, "\delta_A\ox\id"]  \arrow[from=2-3, swap, "\id\ox \delta_B^T"]&\\
C^0_A \ox C^0_B &  & C^1_A \ox C^1_B \\
& C^0_A \ox C^1_B \ar[to=2-1, "\id \ox \delta_B^T"] \ar[to=2-3, swap, "\delta_A\ox \id"]&
\end{tikzcd}
\hspace{0.5 cm}
\begin{tikzcd}[scale cd=0.75]
\mathcal{C}_1 \\
\mathcal{C}_0 \ar[to=1-1, swap, "\partial_{0}"]\\
\mathcal{C}_{-1}\ar[to=2-1, swap, "\partial_{-1}"]
\end{tikzcd}
\end{equation}

\begin{equation*}
\partial_{-1}=\left(\id \ox \delta_B, \quad \delta_A^T \ox \id \right)^T,\\
\partial_0=\left(\delta_A \ox \id, \quad \id \ox \delta_B^T\right)
\end{equation*}
is well defined. In fact:
\begin{align*}
    \partial_{-1} \partial_0 = \delta_A \ox \delta_B^T + \delta_A \ox \delta_B^T = 0.
\end{align*}
Therefore, the complex $\eqref{eq:cab}$ defines a valid CSS code, which we denote by $\cab$ and refer to as the hypergraph product code of the seed matrices $\delta_A$ and $\delta_B$. If the classical code with parity check $\delta_{\ell}, \delta_{\ell}^T$ has parameters $[n_{\ell}, k_{\ell}, d_{\ell}]$ and $[n_{\ell}^T, k_{\ell}^T, d_{\ell}^T]$ respectively ($\ell = A, B$) then the hypergraph product code $\cab$ has parameters:
\begin{align*}
[[n_an_b+n_a^Tn_b^T,\, k_ak_b + k_a^Tk_b^T,\, d_x,\, d_z]],
\end{align*}
where $d_x=\min\{d_a^T, d_b\}$ and $d_z=\min\{d_a, d_b^T\}$, see \cite{bravyi2014homological}.
\subsection{Reshaping of vectors}
\label{sec:reshaping}
One tool we make extensive use of, and from which our decoder takes its name, is the reshaping of vectors of a two-fold tensor product space into matrices (see, for instance, \cite{bravyi2014homological, campbell2019theory}). Consider a basis $\mathcal B$ of the vector space $\ff^{n_1} \ox \ff^{n_2}$:
\begin{equation*}
\mathcal{B}=\{a_i \ox b_j \mid i = 1, \dots, n_1 \text{ and } j = 1, \dots, n_2\}.
\end{equation*}
Then any $v \in \mathbb{F}_2^{n_1} \ox \mathbb{F}_2^{n_2}$ can be written as:
\begin{align*}
    v = \sum_{a_i \ox b_j \in \mathcal {B}} v_{ij} (a_i \ox b_j),
\end{align*}
for some $v_{ij} \in \mathbb F_2$. We call the $n_1 \times n_2$ matrix $V$ with entries $v_{ij}$ the reshaping of the vector $v$.
By this identification, if $\varphi$, $\theta$ are respectively $m_1 \times n_1$ and $m_2 \times n_2$ matrices, then $(\varphi \ox \theta) (V) = \varphi V \theta^T $. The inner product between $u \ox w$ and $v$ in $ \mathbb{F}_2^{n_1} \ox \mathbb{F}_2^{n_2}$ can be computed as 
\begin{align}
\label{eq:inner_product}
   \langle u \ox w, v \rangle = u^T V w.
\end{align}

As we here detail, the identification of operators on the code space $\cab$ with pairs of binary matrices that we used in the main text is rigorously justified by the reshaping of vectors into matrices.
With slight abuse of notation, we refer to binary vectors and binary matrices as operators and vice versa, where the identification is clear via Eq.~\eqref{eq:z_operator}.
\subsection{Graphical representation}
Physical qubits of the code $\cab$ are in one-to-one correspondence with basis elements of the space $\mathcal{C}_0$. If $\{e_{j_a}\}_{j_a}$, $\{e_{j_b}\}_{j_b}$, $\{e_{i_a}\}_{i_a}$, $\{e_{i_b}\}_{i_b}$ are bases of the spaces $C^0_A, C^0_B, C^1_A, C_B^1$ of dimension $n_a, n_b, m_a, m_b$ respectively, then the union of the two sets
\begin{align*}
    \mathcal{B}_L &\coloneqq \{(e_{j_a} \ox e_{j_b}, \, 0)\}\\
    \mathcal{B}_R &\coloneqq \{(0, \, e_{i_a}\ox e_{i_b})\}
\end{align*}
is a basis of $\mathcal{C}_0$. We refer to qubits associated to elements in $\mathcal{B}_L$, or its span, as \emph{left} qubits and to those associated to $\mathcal{B}_R$, or its span, as \emph{right} qubits. Since qubit operators are vectors in $\mathcal{C}_0$, by reshaping, they can be identified with pairs of matrices $(L,\,R)$ where $L$ has size $n_a \times n_b$ and $R$ has size $m_a \times m_b$; in particular, $L$ acts on the left qubits while $R$ acts on the right qubits.

A $Z$-stabilizer for the code associated to the complex \eqref{eq:cab} is any vector in $\im \partial_{-1}$. A generating set for $Z$-stabilizers is:
\begin{equation*}
\mathcal{S}_z \coloneqq \{\partial_{-1}(e_{j_a} \ox e_{i_b})\}_{j_a,i_b}    
\end{equation*}
 where $e_{j_a}$ and $e_{i_b}$ are unit vectors of $C^0_A$ and $C^1_B$ respectively, i.e.\ they are a basis of the two spaces. Let $E_{j_ai_b} \in C^0_A \ox C_B^1$ be the reshaping of $(e_{j_a}\ox e_{i_b})$, i.e.\ it is the matrix with all zeros entries but for the $(j_a,i_b)$-th entry which is 1. The reshape of $\partial_{-1}(e_{j_a} \ox e_{i_b})$ is then given by the pair of matrices:
\begin{equation*}
(L, R) = (E_{j_ai_b}\delta_B,\delta_AE_{j_ai_b}).
\end{equation*}
Logical $Z$-operators are vectors in $\ker \partial_0$ which are not in $\im 
\partial_{-1}$. Specifically, a minimal generating set of logical $Z$-operators is given by \cite{zeng2019higher}:
\begin{align}
\label{eq:logical_z}
\hat{\mathcal{L}}_z&\coloneqq \hat{\mathcal{L}}_z^{\mathrm{left}} \cup \hat{\mathcal{L}}_z^{\mathrm{right}}
\end{align}
where
\begin{align*}
    \hat{\mathcal{L}}_z^{\mathrm{left}} \coloneqq \big\{(k_a \ox e_{j_b}, 0) &: k_a \text{ varies among a basis of } \ker \delta_A,\\ 
    &\quad |e_{j_b}|=1 \text{ and it varies}\\
    &\quad\text{among a basis of } \im(\delta_B^T)^{\bullet}\big\},
\end{align*}
and
\begin{align*}
    \hat{\mathcal{L}}_z^{\mathrm{right}} \coloneqq \big\{(0, e_{i_a}\ox \bar k_b) &: |e_{i_a}| = 1 \text{ and it varies}\\
    &\quad\text{among a basis of } \im (\delta_A)^{\bullet},\\
    &\quad \bar k_b \text{ varies among a basis of } \ker \delta_B^T\big\}.
\end{align*}
The reshaping of vectors in $\hat{\mathcal{L}}_z$ gives the set $\mathcal{L}_z$ of Eq.~\eqref{eq:matrix_logical_z} in the main text. The vector version of logical $X$-operators is likewise obtained from the set of matrices $\mathcal L_x$ of Eq.~\eqref{eq:matrix_logical_x}.

\section{Proofs}
\label{app:proofs}
This Section contains all the proofs of the statements made in the main text.

Broadly speaking, in this work we wanted to characterize $Z$-errors operators on the codespace of $\cab$ associated to the chain complex \eqref{eq:cab}. In order to do so, we first studied the logical $Z$-operators of $\cab$ and introduced a canonical form for them. From homology theory, we know that non-trivial logical $Z$-operators are associated to vectors in $\ker \partial_0$ which do not belong to $\im \partial_{-1}$. Lemma \ref{lemma:ker} below describes all the vectors in $\ker \partial_0$.
\begin{lemma}
\label{lemma:ker}
Let $(L, R) \in \mathcal C_0$ be in $\ker \partial_0$, then:
\begin{align*}
    L &\in \ker \delta_A \ox C_B^0 + C_A^0 \ox \im \delta_B^T,\\
    R &\in C_A^1 \ox \ker \delta_B^T + \im \delta_A \ox C_B^1.
\end{align*}
\end{lemma}
\begin{proof}
Let $(L, R) \in \ker \partial_0$. Then:
\begin{align}
\nonumber
    \partial_0 (L, R) = 0 &\Longleftrightarrow (\delta_A \ox \id) L + (\id \ox \delta_B^T) R = 0,\\
    \label{eq:proof_zero}
    &\Longleftrightarrow \delta_A L + R \delta_B =0.
\end{align}
Eq.~\eqref{eq:proof_zero} yields:
\begin{align}
\delta_A L = R \delta_B = V,
    \label{eq:proof_v}
\end{align}
for some $V \in C_A^{1} \ox C_B^0$. Eq.~\eqref{eq:proof_v} entails that all columns of $V$ belong to $\im \delta_A$ while its rows belong to $\im \delta_B^T$. As a consequence, it must exists $U \in C_A^{0}\ox C_B^{1}$ such that:
\begin{align*}
    V = \delta_A U \delta_B.
    % &= (\delta_A \ox \delta_B^T) U.
\end{align*}
Therefore Eq.~\eqref{eq:proof_v} can be re-written as:
\begin{align*}
    \delta_A L = \delta_A U \delta_B
\end{align*}
which yields:
\begin{align}
    (\delta_A \ox \id)(L + U \delta_B)& = \delta_A L + \delta_A U \delta_B \nonumber\\
    &= V + V\nonumber\\
    &=0 \label{eq:zero_sum}
\end{align}
Equivalently, Eq.~\eqref{eq:zero_sum} states that $L + U \delta_B$ has columns in $\ker \delta_A$:
\begin{align*}
    L + U \delta_B \in \ker \delta_A \ox C_B^0
\end{align*}
and therefore:
\begin{align*}
    L \in \ker \delta_A \ox C_B^0 + C_A^0 \ox \im \delta_B^T,
\end{align*}
as in the thesis. Similarly, we find 
\begin{align*}
   R \in C_A^1 \ox \ker \delta_B^T + \im \delta_A \ox C_B^1. 
\end{align*}
\end{proof}
A proof of Proposition \ref{prop:canonical_form}, reported below for clarity, follows directly combining what said in Appendix \ref{sec:reshaping} and Lemma \ref{lemma:ker}.
\begin{customprop}{\ref{prop:canonical_form}}[Canonical form]
Let $(L, R)$  be a $Z$-operator on the codespace of $\cab$. For a vector space $V \subseteq \ff^n$, we denote by $V^{\bullet}$ any space such that $V \oplus V^{\bullet} \simeq \ff^n$, (see Appendix \ref{app:linear_algebra}). Then, for the operator $(L, R)$, the left part $L$ can be expressed as a sum of a \emph{free part} $M_L$ and a \emph{logical part} $O_L$ such that every row of $M_L$ belongs to $\im \delta_B^T$ and every row of $O_L$ belongs to $(\im \delta_B^T)^{\bullet}$. Similarly, the right part $R$ can be expressed as a sum of a free part $M_R$ and a logical part $O_R$ such that every column of $M_R$ belongs to $\im \delta_A$ and every column of $O_R$ belongs to $(\im \delta_A)^{\bullet}$. Hence, for $(L, R)$ holds:
\begin{align}
\label{eq:cf_2}
\tag{CF}
(L, R) = (M_L + O_L, M_R + O_R).
\end{align}
We refer to the writing given by Eq.~\eqref{eq:cf_2} as \emph{canonical form} of the operator $(L,R)$.
\end{customprop}

In the main text, we have introduced the notions of row-column weight and logical row-column weight for a $Z$-operator on $\cab$. The definition of these two quantities finds its explanation in Proposition \ref{prop:logical-op}, whose proof builds on the results of Lemma \ref{lemma:ker}.
\begin{customprop}{\ref{prop:logical-op}}
If $(L,\,R)$ is a non-trivial logical $Z$-operator of $\cab$ then either $\#\row(L) \ge d_a$ or $\#\col(R)\ge d_b^T$ (or both). 
\end{customprop}
\begin{proof}
% [Proof of Prop.~\ref{prop:logical-op}]
If $(L, R)$ is a non-trivial logical $Z$-operator, it must anti-commute with at least one logical $X$-operator $(L_x, R_x)$. Because a $Z$-operator and a $X$-operator anti-commute if and only if their supports overlap on an odd number of positions, either $L$ and $L_x$ or $R$ and $R_x$ have odd overlap. Without loss of generality, we can assume that the former is verified and we can choose $(L_x, R_x)$ as a left operator of the form 
\begin{align*}
    (L_x, R_x) = (f \ox k, 0),
\end{align*}
where $f$ is a unit vector in $(\im \delta_A^T)^{\bullet}$ and $k \in \ker \delta_B$. In other words, we choose logical $X$-operator $(f\ox k, 0)$ from the set of generators of $X$-logical operators $\hat{\mathcal{L}}_{x}^{\mathrm{left}}$, as in the $X$-version of Eq.~\eqref{eq:logical_z}. The inner product equation for reshaped vectors Eq.~\eqref{eq:inner_product} then yields:
\begin{align}
1 = \big\langle (L_x, 0), (L,\,R)\big\rangle &= \langle L_x, L \rangle + \langle 0, R\rangle \nonumber\\
\label{eq:anticomm}
& =f^TL k .
\end{align}
Now, observe that $(L,\,R) \in \mathcal C_0$ is a non-trivial logical $Z$-operator of $\cab$ if and only if $[L, R]\in \mathcal{H}_0 = \ker \partial_0 / \im \partial_{-1}$ and $[L, R] \neq 0$ or, equivalently, if and only if:
\begin{align*}
  (L,\,R) \in \ker \partial_0\setminus \im \partial_{-1}  .
\end{align*}
In particular, $(L, R)$ belongs to $\ker \partial_0$ and thanks to Lemma \ref{lemma:ker}, we can re-write it as:
\begin{align*}
    (L, R) = (K_A + U_L \delta_B, \bar K_B + \delta_A U_R),
\end{align*}
where columns of $K_A$ belong to $\ker \delta_A$ and rows of $\bar{ K}_B$ belong to $\ker \delta_B^T$. Using Lemma \ref{lemma:ker}'s decomposition for $(L, R) \in \ker \partial_0$, we can expand the matrix-vector product $Lk$ as:
\begin{align}
Lk &= (K_A + U_L \delta_B) k \nonumber\\
&= K_A k + U_L\delta_B k& \nonumber\\
\label{eq:column_in_ker}
&=K_A k &\text{since } k \in \ker \delta_B.
\end{align}
Eq.~\eqref{eq:column_in_ker} entails $Lk=K_A k$ and therefore that $Lk$, being a linear combination of column-vectors in $\ker \delta_A$, belongs to $\ker \delta_A$ itself. Furthermore, by Eq.~\eqref{eq:anticomm}, $Lk \neq 0$. To sum up, $Lk$ is a non-zero vector in $\ker \delta_A$ and therefore it must have Hamming weight at least $d_a$. As a consequence, $L$ is a matrix with at least $d_a$ rows:
\begin{align*}
\#\row(L) \ge d_a.
\end{align*}
Similarly, we would have found:
\begin{align*}
\#\col(R) \ge d_b^T,
\end{align*}
if we had assumed that $(L,R)$ anti-commuted with a logical $X$-operator $(0, R_x)$ in $\hat{\mathcal{L}}_x^{\mathrm{right}}$.
\end{proof}
Corollary \ref{cor:logical} follows easily. 
\begin{customcorollary}{\ref{cor:logical}}

If $(L, R)$ is a non-trivial logical $Z$-operator on $\cab$, at least one of the following hold:
\begin{enumerate}[label=(\roman*)]
\item \label{pt:cor_logical_one} $L$ has at least $d_a$ rows which are not in $\im \delta^T_B$ when seen as vectors in $C^0_B$.
\item \label{pt:cor_logical_two} $R$ has at least $d_b^T$ columns which are not in $\im \delta_A$ when seen as vectors of $C^1_A$. 
\end{enumerate}
\end{customcorollary}

\begin{proof}
Write $(L, R)$ in its canonical form:
\begin{align*}
    (L, R)= (M_L + O_L, M_R + O_R),
\end{align*} 
and let 
\begin{align*}
    M_L &= N_L \delta_B,\\
    M_R &= \delta_A N_R,
\end{align*}
for some binary matrices $N_L, N_R$ of size $n_a \times m_b$. As done in the proof of Proposition \ref{prop:logical-op}, consider a logical $X$-operator $(f \ox k, 0)$ such that it anti-commutes with $(L, R)$. 
Combining the canonical form of $L$ and Eq.~\eqref{eq:column_in_ker}, yields:
\begin{align*}
    Lk &= (O_L + N_L \delta_B)k \\
    &= O_L k &\text{since }k \in \ker \delta_B\\
    &= K_A k &\text{by Eq.~\eqref{eq:column_in_ker}}
\end{align*}
for some $n_a \times n_b$ matrix $K_A$ with columns in $\ker \delta_A$. By the same argument used in the proof of Proposition \ref{prop:logical-op}, we find:
\begin{align*}
    |K_A k| \ge d_A \Rightarrow |O_L k| \ge d_A,
\end{align*}
and in particular that $O_L$ has at least $d_a$ non-zero rows. Since by definition of canonical form the non-zero rows of $O_L$ are precisely those rows of $L$ which do not belong to $\im \delta_B^T$, we have proven point \ref{pt:cor_logical_one}. Point \ref{pt:cor_logical_two} follows similarly in the case $(L, R)$ anti-commutes with at least one logical $X$-operator of the form $(0, R_x)$.
\end{proof}
Corollary \ref{cor:logical}, together with Proposition \ref{prop:invariant} below, justifies the definition of the logical row-column weight for $Z$-operators on $\cab$ (Definition \ref{definition:logical_weight}). The logical row-column weight of $(L, R)$ is denoted by the symbol $\wtrcl(L, R)$ and stands for the integer pair $(\#\rowl(L),\, \#\coll(R))$ where $\#\rowl(L)$ is the number of rows of $L$ that are not in $\im \delta_B^T$ and $\#\coll(R)$ is the number of columns of $R$ which are not in $\im \delta_A$. Proposition \ref{prop:invariant}, that we now prove, states that the logical row-column weight of a $Z$-operator on $\cab$ is an homology invariant of the chain complex \eqref{eq:cab} and therefore it legitimates the name choice for this quantity. 
\begin{customprop}{\ref{prop:invariant}}
The logical row-column weight of a $Z$-operator on $\cab$ is an invariant of its homology class.
\end{customprop}
\begin{proof}
Let $[L, R]$ be the homology class of $(L,\,R)$: 
\begin{align*}
[L,\, R] = \big\{(L + G_L,\, R+G_R) : &\,(G_L, \,G_R)\big. \\
&\big.\text{ is a $Z$-stabiliser}\big\}.
\end{align*} 
The operator $(G_L, G_R) \in \mathcal{C}_0$ is a $Z$-stabilizer for $\cab$ if and only if 
\begin{align}
    (G_L, G_R) &= \partial_{-1}(U) \nonumber\\
    &=(U\delta_B, \delta_A U) \label{eq:stab}
\end{align}
For some $n_a \times m_b$ binary matrix $U$. Eq.~\eqref{eq:stab} entails that any row of $G_L$ belongs to $\im \delta_B^T$ and any column of $G_R$ belongs to $\im \delta_A$. Therefore, if we write $(L, R)$ in its canonical form:
\begin{align*}
    (L, R) = (M_L + O_L, M_R + O_R),
\end{align*}
we see that we can `delete' all the rows of $M_L$ by adding a stabiliser and hence `move' part of the support of the operator $(L, R)$ from the left qubits to the right qubits. Specifically, if $M_L = N_L \delta_B$ for some $n_a \times m_b$ binary matrix $N_L$, we consider the stabiliser $G = (N_L \delta_B, \delta_A N_L)$ and we obtain:
\begin{align*}
    (L, R) + G = (O_L, M_R + O_R + \delta_A N_L).
\end{align*}
Similarly, we could move the $M_R$ part of the operator $(L, R)$ from the right qubits to the left qubits, by adding the stabilizer $G' = (N_R \delta_B, \delta_A N_R )$, for a $n_a \times m_b$ matrix $N_R$ such that $M_R = \delta_A N_R$. 

On the other hand though, it is not possible to delete non-zero rows of $O_L$ via stabiliser addition. In other words, it is not possible to remove, via stabiliser addition, any of the rows of $L$ that are not in $\im \delta_B^T$. Hence, the number $\#\rowl(L)$ of non-zero rows of $O_L$ is an homology invariant. Likewise, we find that it is not possible to delete any column in $O_R$ by adding stabilisers and therefore $\#\coll(R)$ is a logical invariant too.
\end{proof} 
The proof of Proposition \ref{prop:invariant} actually entails a stronger result than the invariance of the row-column weight of $Z$-operators on $\cab$. Namely, we have proven that the indices of the rows and the columns in the sets $\rowl$ and $\coll$ respectively, are homology invariants of the reshaped $Z$-operators $(L, R)$ on $\cab$. However, because to prove the correctness of ReShape it is sufficient to look at the cardinality of the two sets $\rowl$ and $\coll$, we decided to state Proposition \ref{prop:invariant} in this more compact and elegant form.

We can now prove Proposition \ref{prop:decoder}.
\begin{customprop}{\ref{prop:decoder}}
Let $S$ be a $X$-syndrome matrix for $\cab$ and $(L, R)$ any valid solution to the Syndrome Equation
\begin{align*}
    \tag{\ref{eq:se}}\sigma(L,\,R)=S.
\end{align*}
Suppose that the minimum weight operator $(L_{\min},\, R_{\min})$ with syndrome $S$ has $(d_a/2, d_b^T/2)$-bounded logical row-column weight i.e.\:
\begin{align*}
    \wtrcl(L_{\min},\, R_{\min}) =(\#\rowl(L_{\min}), \#\coll(R_{\min})),
\end{align*}
is such that
\begin{align}
\tag{\ref{eq:teo_weight}}
    \#\rowl(L_{\min}) < \frac{d_a}{2} \quad\text{ and }\quad
    \#\coll(R_{\min}) < \frac{d_b^T}{2}.
\end{align}
Then, on input $\mathscr{D}_{\delta_A}$, $\mathscr{D}_{\delta_B^T}$, $S$ and $(L,\,R)$, ReShape outputs a correct solution $(\tilde L,\, \tilde R)$ of \eqref{eq:se}, provided that the classical decoders $\mathscr{D}_{\delta_A}$, $\mathscr{D}_{\delta_B^T}$ succeed. In other words, the solution $(\tilde L, \tilde R)$ found by ReShape is in the same homology class as the minimum weight operator with syndrome $S$: 
\begin{align*}
    [L_{\min}, \, R_{\min}] = [\tilde L,\, \tilde R].
\end{align*}
\end{customprop}
\begin{proof} 
This is a proof by contradiction: we suppose that the minimum weight solution and the solution found by ReShape (Algorithm \ref{algo:reshape}) are not homologically equivalent and we find as a consequence that the minimum weight solution need to have high logical row-column weight.

Let $(L, R)$ be the valid solution of \eqref{eq:se} in input to ReShape and $(\tilde L,\, \tilde R)$ be the recovery operator found. 

First note that $\sigma(\tilde L, \tilde R) = \sigma(L, R)$. In fact, the Split step only finds the canonical form of $(L, R)$ and therefore changes neither the operator $(L, R)$ nor its syndrome. The Decode step, possibly adds to $(L, R)$ logical $Z$-operators $(L_z, R_z)$ such that $\sigma(L_z, R_z) = 0$ and therefore, even when it changes the operator, it preserves its syndrome. 

Suppose now that the solution found by ReShape and the minimum weight solution $(L_{\min}, R_{\min})$ of \eqref{eq:se} belong to two different homology classes:
\begin{align*}
    [\tilde L ,\, \tilde R] \neq [L_{\min},\, R_{\min}],
\end{align*}
where:
\begin{align*}
    \#\rowl(L) < \frac{d_a}{2} \quad\text{ and }\quad
    \#\coll(R) < \frac{d_b^T}{2}.
\end{align*}
Since both $(L_{\min}, R_{\min})$ and $(\tilde L, \tilde R)$ are valid solution of \eqref{eq:se}, they must differ for an operator with zero $X$-syndrome. Because $(L_{\min}, R_{\min})$ and $(\tilde L, \tilde R)$ are not homologically equivalent, they must differ for a non-trivial $Z$-operator in the normaliser $\mathcal N(\mathcal{S})$ of the stabiliser group. As such, they must differ for an operator which is the sum of a $Z$-stabiliser and a non-trivial logical operator:
\begin{align}
\label{eq:diff_solutions}
    (L_{\min}, R_{\min}) = (\tilde L, \tilde R) + (G_L, G_R) + (L_z, R_z),
\end{align}
where $(G_L, G_R)$ is a $Z$-stabiliser and $(L_z, R_z)$ is a non-trivial logical $Z$-operator.

Without loss of generality we assume that $(L_z, R_z)$ is non-trivial on the left qubits, meaning that $L_z$ has at least one non-zero column in $\ker \delta_A$. The proof is substantially the same in case it is non-trivial on the right qubits. 

First, write the left operators $L_{\min}$ and $\tilde L$ in their canonical form with respect to the same unit-vector basis used to write the logical operators in $\mathcal L_z$ (see Eq.~\eqref{eq:nice_basis}):
\begin{align*}
    L_{\min} &= M_{\min} + O_{\min},\\
    \tilde L &= \tilde M + \tilde O.
\end{align*}
Note that, by construction, the left operator $L_z + G_L$ is already in its canonical form, where $L_z$ is its logical part and $G_L$ is its free part. By Eq.~\eqref{eq:splitting_spaces}, the sum is direct and therefore the equality given by Eq.~\eqref{eq:diff_solutions} must hold component-wise for the free part and the logical part:
\begin{align}
    M_{\min} &= \tilde M + G_L \nonumber\\
    \label{eq:logical_part}
    O_{\min} &= \tilde O + L_z.
\end{align}
Let now focus on the logical part equality expressed by Eq.~\eqref{eq:logical_part} and let $L_z^j$ be a non-zero column of $L_z$ in $\ker \delta_A$. Then:
\begin{align}
\label{eq:colum_logical_part}
    O_{\min}^j = \tilde O^j + L_z^j, \quad L_z^j \in \ker{\delta_A}.
\end{align}
Eq.~\eqref{eq:decode_min} for the classical decoder $\mathscr{D}_{\delta_A}$, entails:
\begin{align*}
   |\tilde O^j |= \min_{k \in \ker \delta_A} |v + k|  
\end{align*}
for some input vector $v$ defined by $L$. In particular, no vector $k' \in \ker \delta_A$ can overlap with $\tilde O^j$ in more than $d_A/2$ positions, otherwise we would have $|v + k'| < |\tilde O^j|$, against the assumption that $|\tilde O^j|$ is minimum. Thanks to this observation and considering the Hamming weight of the terms in Eq.~\eqref{eq:colum_logical_part}, we obtain:
\begin{align*}
    \lvert O_{\min}^j\rvert &= \lvert \tilde O^j + L_z^j \rvert\\
    &= |\tilde O^j| + |L_z^j| - 2 |\tilde O^j \wedge L_z^j|\\
    & \ge \lvert L_z^j\rvert - \lvert \tilde O^j \wedge L_z^j \rvert\\
    &\ge d_a - \frac{d_a}{2} = \frac{d_a}{2}, &\text{by Eq.~\eqref{eq:colum_logical_part}}.
\end{align*}
Because the weight of any of the columns of a matrix is a lower bound on the number of its non-zero rows, we have:
\begin{align*}
    \#\row(O_{\min}) \ge \frac{d_a}{2}.
\end{align*}
By definition of canonical form, this is equivalent to:
\begin{align*}
    \#\rowl(L_{\mathrm{min}}) \ge \frac{d_a}{2},
\end{align*}
against the assumption. 

We stress that the number $\#\row(O_{\min})$ of rows of $L_{\min}$ which do not belong to $\im \delta_B^T$, does not depend on the particular splitting chosen for the canonical form. In fact, as stated in Proposition \ref{prop:invariant}, the logical row weight of the left part of a $Z$-operator is an homology invariant. An argument similar to the one just outlined for the left part of $(L_{\min}, R_{\min})$ holds for its right part and yields:
\begin{align*}
   \#\coll(R_{\min}) \ge d_b^T/2, 
\end{align*}
again contradicting the assumption. In conclusion, we have reached a contradiction and therefore it must be:
\begin{align*}
    [L_{\min}, R_{\min}] = [\tilde L, \tilde R].
\end{align*}
\end{proof}

\section*{Acknowledgement}
We thank Christophe Vuillot for helpful discussions and for carefully reading a draft of this work. AOQ thanks Joschka Roffe for providing the matrices used in the simulations. 
ETC's contributions were made while he was at the University of Sheffield. 
\bibliographystyle{IEEEtran}
\bibliography{ReShapeLib}
\end{document}